\documentclass[]{siamart190516}
\numberwithin{equation}{section}
\usepackage{amsmath,amssymb}
\usepackage{stmaryrd}
\SetSymbolFont{stmry}{bold}{U}{stmry}{m}{n}
\usepackage{arydshln}
\usepackage{appendix}
\usepackage[caption=false]{subfig}

\usepackage{cleveref}
\allowdisplaybreaks

\def \i {{\mathrm{i}}}

\def \kr {k_\rho}
\newcommand{\J}[1]{\bd{J}_{#1}}
\newcommand{\bd}[1]{\mathbf{#1}}
\newcommand{\bs}[1]{\boldsymbol{#1}}

\def \lbr {\llbracket}
\def \rbr {\rrbracket}

\def \gu {\bd{g}_{\bd{u}}}

\def \gur {\widehat{\bd{g}}_{\bd{u}}^{\mathrm{r}}}
\def \guh {\widehat{\bd{g}}_{\bd{u}}}

\newsiamremark{remark}{Remark}

\title{A Matrix Basis Formulation for the Green's Functions of Maxwell's Equations and the Elastic Wave Equations in Layered Media}
\author{ Wenzhong Zhang\thanks{Department of Mathematics, Southern Methodist University, Dallas, TX 75275.}
        \and Bo Wang\thanks{LCSM(MOE), School of Mathematics and Statistics, Hunan Normal University, Changsha, Hunan, 410081, P. R. China.
Department of Mathematics, Southern Methodist University, Dallas, TX 75275. This author acknowledges the financial support provided by NSFC (grant 11771137),
the Construct Program of the Key Discipline in Hunan Province and a Scientific Research Fund of Hunan Provincial Education Department (No. 16B154).}
        \and Wei Cai\thanks{Corresponding author, Department of Mathematics, Southern Methodist University, Dallas, TX 75275({\tt cai@smu.edu}).}
        }

\begin{document}
\maketitle

\begin{abstract}
A matrix basis formulation is introduced to represent the $3 \times 3$ dyadic Green's functions in the frequency domain for the Maxwell's equations and the elastic wave equation in layered media.
The formulation can be used to decompose the Maxwell's Green's functions into independent TE and TM components, each satisfying a Helmholtz equation, and decompose the elastic wave Green's function into the S-wave and the P-wave components.
In addition, a derived vector basis formulation is applied to the case for acoustic wave sources from a non-viscous fluid layer.
\end{abstract}

\section{Introduction}
The layered media dyadic Green's functions (LMDG) for the Maxwell's equations and the elastic wave equations are studied and used in the integral equation solvers for wave fields \cite{chew, cai2013em, duan2018}.
These Green's functions are $3 \times 3$ tensors, governed by the acoustic wave equations and their variants, with certain interface conditions across boundaries of the layers.
In addition, an acoustic wave originating from a source in a non-viscous fluid will be transmitted into elastic waves in a neighboring solid layer. In this case, we will use the Green's function in terms of pressure for the fluid region related to a compression P-wave while a 3-dimensional displacement vector will be used for the dyadic Green's function in the solid layer where an additional S-wave also propagates.

A naive direct derivation of these Green's functions takes all the $9$ entries of the $3 \times 3$ dyadic into consideration for each layer, while the interface conditions will tangle all the entries together.
However, in fact some of the entries are linearly dependent or even identical.
For the Maxwell's LMDG, a number of formulations have been proposed to simplify the derivation, such as the Sommerfeld potential \cite{sommerfeld1964}, the transverse potential \cite{erteza1969,michalski1987}, the Michalski-Zheng formulations \cite{michalski1990}, the $E_z$-$H_z$ formulation \cite{kong1972, chew2006matrix}, etc.
The Sommerfeld potential and the transverse potential reduce the number of entries to be calculated to $5$ while
the $E_z$-$H_z$ approach uses merely $2$ scalar terms, based on the TE/TM mode decomposition.
For the elastic wave equation, the dyadic Green's function for the half-space problem was discussed in \cite{gosling1994}.

The purpose of this paper is to present a new matrix representation of the $3\times 3$ tensor Green's functions using a linear matrix basis, providing an alternative point of view of the previously known approaches.
The matrix basis representation gives a general formulation for the Green's functions of both the Maxwell's equations and the elastic wave equations.
Mathematical theories are then developed to justify the representation in both cases.
There are several remarkable benefits resulting from the matrix basis formulation (MBF).
First, the coefficients of the matrix basis are all rotationally symmetric in the horizontal directions, so that the evaluation of the reflection/transmission coefficients in the layers are simplified.
Second, the Maxwell's Green's functions can be naturally decomposed into independent TE and TM components within this formulation, leading to the 2-term $E_z$-$H_z$ result,  but also with a clearer explanation about the interface treatment. Meanwhile, the elastic wave Green's function is decomposed into S-wave components and P-wave components by matrix rows.
Third, the rotational symmetry allows us to apply fast solvers easily, e.g. the fast multipole method in layered media \cite{bo2019hfmm}.
We also develop a vector basis formulation which is simplified from the matrix version, used for the LMDG of the mix-phase elastic wave equations where the source originates in fluid medium.


The rest of this paper is organized as follows.
In \Cref{sect-2} we establish the theories of the matrix basis, propose the formulation, and provide guidelines about how the theories and the formulation are applied to the LMDG of the Maxwell's equations and the elastic wave equations.
In \Cref{sect-3} the details of the Maxwell's Green's functions in layered media are explained, including a 5-term matrix-based general formulation and the concise 2-term formulation.
In \Cref{sect-4} we discuss the layered elastic wave Green's functions.
A 5-term formulation separating the S-wave and the P-wave by definition is proposed.
Then, the result is generalized to the vector case for sources in fluid media.
The mixed interface conditions between different medium phase types are enumerated in details.

\section{The matrix basis formulation}\label{sect-2}
In this section we set up the matrix basis used for the Green's functions of the Maxwell's equations or the elastic wave equation, and develop basic theories of the basis coefficients.

Suppose in both problems the domain has a total of $L+1$ layers, indexed by $0, \cdots, L$ from top to bottom, separated by the interfaces $z=d_l$, $0 \le l \le L-1$ arranged from large to small.
Suppose $k_i$, $1 \le i \le I$, consist all the wave numbers in these layers.
Note that there are distinct elastic wave numbers for the S-wave and the P-wave in each solid layer.
The interaction between the fixed source $\bd{r}' = (x',y',z')$ and the target $\bd{r}=(x,y,z)$ is studied, assuming both of them do not locate on the interfaces.

Take the 2-D Fourier transform from $(x-x', y-y')$ to $(k_x, k_y)$:
\begin{equation}\label{eq-2DFT}
f(x,y) = \frac{1}{4\pi^2} \iint_{\mathbb{R}^2} e^{\i k_x (x - x') + \i k_y(y - y')} \widehat{f}(k_x, k_y) dk_x dk_y.
\end{equation}
Let $(\kr, \alpha)$ be the polar coordinates of $(k_x, k_y)$.
Let the field $\mathbb{F}_0$ be the field extension from $\mathbb{C}$ with certain functions of $k_x$ and $k_y$
\begin{equation}\label{def-F0}
\mathbb{F}_0 = \mathbb{C} \left( \i \sqrt{k_i^2 - \kr^2}, e^{\i \sqrt{k_i^2 - \kr^2} d_j}, e^{\i \sqrt{k_i^2 - \kr^2} z}, e^{\i \sqrt{k_i^2 - \kr^2} z'}: 1 \le i \le I, 0 \le j \le L-1 \right),
\end{equation}
where other variables are treated as given constants.
Obviously, functions in $\mathbb{F}_0$ do not depend on $\alpha$.
Note that $\kr^2 \in \mathbb{F}_0$ because
\begin{equation*}
\kr^2 = k_i^2 + \left( \i \sqrt{k_i^2 - \kr^2}\right)^2.
\end{equation*}
Then, define the field $\mathbb{F}$ by the 2-term extension
\begin{equation}
\mathbb{F} = \mathbb{F}_0(\i k_x, \i k_y).
\end{equation}

\begin{remark}
The coordinates $z$ and $z'$ are indeed redundant in the definition of $\mathbb{F}_0$ for the matrix basis theory.
They are included only for the convenience of statements in the following sections.
\end{remark}

One of our anticipations on the matrix basis is to represent the tensor Green's functions with all the coefficients in $\mathbb{F}_0$, i.e. the information of the polar angle is only kept in the matrix basis.
For this purpose, based on the observation on calculated formulas \cite{cho2017}, we propose the matrix basis $\J{1}, \cdots, \J{9}$ in the frequency domain as follows.
\begin{proposition}[The matrix basis]\label{prop-mat-basis}These matrices form a basis of $\mathbb{F}^{3 \times 3}$:
\begin{align}\label{mat-J}
\begin{split}
\J{1} & = \begin{bmatrix}
1 & & \\
& 1 & \\
& & 0
\end{bmatrix}, \enspace
\J{2}  = \begin{bmatrix}
0 & & \\
& 0 & \\
& & 1
\end{bmatrix}, \enspace
\J{3}  = \begin{bmatrix}
0 & 0 & \i k_x \\
0 & 0 & \i k_y \\
0 & 0 & 0
\end{bmatrix}, \\
\J{4} & = \begin{bmatrix}
0 & 0 & 0 \\
0 & 0 & 0 \\
\i k_x & \i k_y & 0
\end{bmatrix}, \enspace
\J{5}  = \begin{bmatrix}
-k_x^2 & -k_x k_y & 0 \\
-k_x k_y & -k_y^2 & 0 \\
0 & 0 & 0
\end{bmatrix}, \enspace
\J{6}  = \begin{bmatrix}
0 & 0 & 0 \\
0 & 0 & 0 \\
-\i k_y & \i k_x & 0
\end{bmatrix}, \\
\J{7} & = \begin{bmatrix}
0 & 0 & \i k_y \\
0 & 0 & -\i k_x \\
0 & 0 & 0
\end{bmatrix}, \enspace
\J{8}  = \begin{bmatrix}
k_x k_y & k_y^2 & 0 \\
-k_x^2 & -k_x k_y & 0 \\
0 & 0 & 0
\end{bmatrix}, \enspace
\J{9}  = \begin{bmatrix}
0 & 1 & 0 \\
-1 & 0 & 0 \\
0 & 0 & 0
\end{bmatrix}.
\end{split}
\end{align}
\end{proposition}
The proof is trivial.

The matrix basis deserves to be divided into two groups once the products between them are studied.
For any subfield $\mathbb{K} \subset \mathbb{F}$, define vector spaces
\begin{align}
\begin{split}
\mathfrak{R}(\mathbb{K})&= \mathrm{span}_{\mathbb{K}}(\J{1},\cdots,\J{5}), \\
\mathfrak{I}(\mathbb{K})&= \mathrm{span}_{\mathbb{K}}(\J{6},\cdots,\J{9}), \\
\mathfrak{M}(\mathbb{K})&= \mathrm{span}_{\mathbb{K}}(\J{1},\cdots,\J{5},\J{6},\cdots,\J{9}).
\end{split}
\end{align}

\begin{proposition}\label{prop-mat-ring}
Let $\mathbb{K}$ be any subfield of $\mathbb{F}$ containing $\kr^2$.
Then,
\begin{itemize}
\item $\mathfrak{M}(\mathbb{K}) = \mathfrak{R}(\mathbb{K}) \oplus \mathfrak{I}(\mathbb{K})$ is the direct sum.
\item $\mathfrak{R}(\mathbb{K}), \mathfrak{M}(\mathbb{K})$ are rings with matrix addition and matrix multiplication.
\end{itemize}
\end{proposition}
\begin{proof}
The direct sum is obvious.
For the ring property, notice the identity matrix
\begin{equation}
\bd{I} = \J{1} + \J{2} \in \mathfrak{R}(\mathbb{K}) \subset \mathfrak{M}(\mathbb{K}),
\end{equation}
and the product table of the matrices $\J{1},\cdots,\J{9}$
\begin{align}\label{eq-matrix-product-table}
\begin{split}
{}&\begin{bmatrix}
\J{1}^T & \cdots & \J{9}^T
\end{bmatrix}^T \cdot \begin{bmatrix}
\J{1} & \cdots & \J{9}
\end{bmatrix} = \\
& \left[\begin{array}{ccccc:cccc}
\J{1} & \bd{0} & \J{3} & \bd{0} & \J{5} & \bd{0} & \J{7} & \J{8} & \J{9} \\
\bd{0} & \J{2} & \bd{0} & \J{4} & \bd{0} & \J{6} & \bd{0} & \bd{0} & \bd{0} \\
\bd{0} & \J{3} & \bd{0} & \J{5} & \bd{0} & \J{8}-\kr^2 \J{9} & \bd{0} & \bd{0} & \bd{0} \\
\J{4} & \bd{0} & -\kr^2\J{2} & \bd{0} & -\kr^2\J{4} & \bd{0} & \bd{0} & \bd{0} & \J{6} \\
\J{5} & \bd{0} & -\kr^2\J{3} & \bd{0} & -\kr^2\J{5} & \bd{0} & \bd{0} & \bd{0} & \J{8}-\kr^2 \J{9} \\ \hdashline
\J{6} & \bd{0} & \bd{0} & \bd{0} & \bd{0} & \bd{0} & \kr^2 \J{2} & -\kr^2 \J{4} & -\J{4} \\
\bd{0} & \J{7} & \bd{0} & -\J{8} & \bd{0} & \kr^2 \J{1}+\J{5} & \bd{0} & \bd{0} & \bd{0} \\
\J{8} & \bd{0} & \kr^2\J{7} & \bd{0} & -\kr^2\J{8} & \bd{0} & \bd{0} & \bd{0} & -\kr^2 \J{1}-\J{5} \\
\J{9} & \bd{0} & \J{7} & \bd{0} & -\J{8} & \bd{0} & -\J{3} & \J{5} & -\J{1}
\end{array}\right]
\end{split}
\end{align}
which ensures the matrix multiplication is closed in either $\mathfrak{M}(\mathbb{K})$ or $\mathfrak{R}(\mathbb{K})$.
\end{proof}
The product table \cref{eq-matrix-product-table} immediately implies the product rules below.
\begin{proposition}[The product rules]\label{prop-mat-product-1}
Let $\mathbb{K}$ be any subfield of $\mathbb{F}$ containing $\kr^2$.
\begin{itemize}
\item If $ {\bd{A}} \in \mathfrak{R}(\mathbb{K})$, $ {\bd{B}} \in \mathfrak{R}(\mathbb{K})$, then $ {\bd{A}}\cdot {\bd{B}} \in \mathfrak{R}(\mathbb{K})$.
\item If $ {\bd{A}} \in \mathfrak{R}(\mathbb{K})$, $ {\bd{B}} \in \mathfrak{I}(\mathbb{K})$, then $ {\bd{A}}\cdot {\bd{B}} \in \mathfrak{I}(\mathbb{K})$.
\item If $ {\bd{A}} \in \mathfrak{I}(\mathbb{K})$, $ {\bd{B}} \in \mathfrak{R}(\mathbb{K})$, then $ {\bd{A}}\cdot {\bd{B}} \in \mathfrak{I}(\mathbb{K})$.
\item If $ {\bd{A}} \in \mathfrak{I}(\mathbb{K})$, $ {\bd{B}} \in \mathfrak{I}(\mathbb{K})$, then $ {\bd{A}}\cdot {\bd{B}} \in \mathfrak{R}(\mathbb{K})$.
\end{itemize}
\end{proposition}
The behavior resembles the real numbers and the imaginary numbers, which is why the letters $\mathfrak{R}$ and $\mathfrak{I}$ are used here.

\begin{definition}[The matrix basis formulation]\label{def-MBF}Define the linear space
\begin{equation}\label{eq-R0}
\mathfrak{R}^0 = \mathrm{span}_{\mathbb{F}_0}\{\J{1},\cdots,\J{5}\} = \left\{ \sum_{j=1}^5 a_j \J{j}: a_j \in \mathbb{F}_0 \right\}.
\end{equation}
The linear expansion of functions in $\mathfrak{R}^0$ with basis $\J{1}, \cdots, \J{5}$ by definition is called the matrix basis formulation.
\end{definition}

In future sections we will claim that the matrix basis formulation can be used to efficiently solve the Green's functions for the Maxwell's equations and the elastic wave equation in each layer.
Since the tensors are across the layers and to be solved together, theories of the block matrices are also necessary.

For any subfield $\mathbb{K} \subset \mathbb{F}$ and any $p,q \in \mathbb{N}$, define the linear spaces of block matrices
\begin{align}\label{eq-block-matrices}
\begin{split}
\mathfrak{M}_{p \times q}(\mathbb{K}) &= \left\{ \sum_{j=1}^{9} \bd{K}_j \otimes \J{j} : \bd{K}_j \in \mathbb{K}^{p \times q}, 1 \le j \le 9 \right\}, \\
\mathfrak{R}_{p \times q}(\mathbb{K}) &= \left\{ \sum_{j=1}^{5} \bd{K}_j \otimes \J{j} : \bd{K}_j \in \mathbb{K}^{p \times q}, 1 \le j \le 5 \right\}, \\
\mathfrak{I}_{p \times q}(\mathbb{K}) &= \left\{ \sum_{j=6}^{9} \bd{K}_j \otimes \J{j} : \bd{K}_j \in \mathbb{K}^{p \times q}, 6 \le j \le 9 \right\},
\end{split}
\end{align}
where $\otimes$ is the Kronecker product.
It is straightforward that in the above definitions the decompositions are unique (since in each $3 \times 3$ block it's unique), and that
\begin{equation}
\mathfrak{M}_{p \times q}(\mathbb{K}) = \mathfrak{R}_{p \times q}(\mathbb{K}) \oplus \mathfrak{I}_{p \times q}(\mathbb{K})
\end{equation}
is the direct sum.
Moreover, the product rules are easily generalized to block matrices.
\begin{proposition}[The product rules for block matrices]\label{lemma-product-rule-block}Let $p,q,r \in \mathbb{N}$.
Let $\mathbb{K}$ be any subfield of $\mathbb{F}$ containing $\kr^2$.
\begin{itemize}
\item If $\bar{\bd{A}} \in \mathfrak{R}_{p\times r}(\mathbb{K})$, $\bar{\bd{B}} \in \mathfrak{R}_{r\times q}(\mathbb{K})$, then $\bar{\bd{A}}\cdot\bar{\bd{B}} \in \mathfrak{R}_{p\times q}(\mathbb{K})$.
\item If $\bar{\bd{A}} \in \mathfrak{R}_{p\times r}(\mathbb{K})$, $\bar{\bd{B}} \in \mathfrak{I}_{r\times q}(\mathbb{K})$, then $\bar{\bd{A}}\cdot\bar{\bd{B}} \in \mathfrak{I}_{p\times q}(\mathbb{K})$.
\item If $\bar{\bd{A}} \in \mathfrak{I}_{p\times r}(\mathbb{K})$, $\bar{\bd{B}} \in \mathfrak{R}_{r\times q}(\mathbb{K})$, then $\bar{\bd{A}}\cdot\bar{\bd{B}} \in \mathfrak{I}_{p\times q}(\mathbb{K})$.
\item If $\bar{\bd{A}} \in \mathfrak{I}_{p\times r}(\mathbb{K})$, $\bar{\bd{B}} \in \mathfrak{I}_{r\times q}(\mathbb{K})$, then $\bar{\bd{A}}\cdot\bar{\bd{B}} \in \mathfrak{R}_{p\times q}(\mathbb{K})$.
\end{itemize}
\end{proposition}

The following theorem will lead to the main result of this section.
\begin{theorem}[Solution filtering]\label{thm-solution-filtering}Suppose $p,q,r \in \mathbb{N}$, the block matrices $\bar{\bd{A}} \in \mathfrak{R}_{p \times r}(\mathbb{F}_0)$, $\bar{\bd{X}} \in \mathfrak{M}_{r \times q}(\mathbb{F})$ and $\bar{\bd{B}} \in \mathfrak{R}_{p \times q}(\mathbb{F}_0)$ satisfy $\bar{\bd{A}} \cdot \bar{\bd{X}} = \bar{\bd{B}}$.
Then, there exists a ``filtered'' block matrix $\bar{\bd{X}}_0 \in \mathfrak{R}_{r \times q}(\mathbb{F}_0)$, i.e. each block of $\bar{\bd{X}}_0$ has the matrix basis representation, so that $\bar{\bd{A}} \cdot \bar{\bd{X}}_0 = \bar{\bd{B}}$.
\end{theorem}
\begin{proof}
We filter the solution from $\mathfrak{M}_{r \times q}(\mathbb{F})$ to $\mathfrak{R}_{r \times q}(\mathbb{F}_0)$ with an intermediate step in $\mathfrak{R}_{r \times q}(\mathbb{F})$.

First, write the direct sum decomposition $\bar{\bd{X}} = \bar{\bd{X}}_1 \oplus \bar{\bd{X}}_2$, where $\bar{\bd{X}}_1 \in \mathfrak{R}_{r \times q}(\mathbb{F})$ and $\bar{\bd{X}}_2 \in \mathfrak{I}_{r \times q}(\mathbb{F})$.
By \cref{lemma-product-rule-block} we immediately get $\bar{\bd{A}}\cdot \bar{\bd{X}}_1 \in \mathfrak{R}_{p \times q}(\mathbb{F})$ and $\bar{\bd{A}}\cdot \bar{\bd{X}}_2 \in \mathfrak{I}_{p \times q}(\mathbb{F})$, so
\begin{equation}
\bar{\bd{A}}\cdot \bar{\bd{X}}_1 + \bar{\bd{A}}\cdot \bar{\bd{X}}_2
\end{equation}
is the direct sum decomposition of $\bar{\bd{B}} \in \mathfrak{M}_{p \times q}(\mathbb{F})$.
Therefore $\bar{\bd{A}}\cdot \bar{\bd{X}}_1 = \bar{\bd{B}}$.

Then, let $$\bar{\bd{A}} = \sum_{j=1}^{5}\bd{A}_j \otimes \bd{J}_j,\quad \bar{\bd{X}}_1 = \sum_{j=1}^{5}\bd{X}_j^1 \otimes \bd{J}_j, \quad \bar{\bd{B}} = \sum_{j=1}^{5}\bd{B}_j \otimes \bd{J}_j $$
where each $\bd{A}_j \in \mathbb{F}_0^{p \times r}$, $\bd{X}_j^1 \in \mathbb{F}^{r \times q}$ and $\bd{B}_j \in \mathbb{F}_0^{p \times q}$.
When treating $\bd{X}_j^1$ as the solution to the linear equation $\bar{\bd{A}}\cdot\bar{\bd{X}}_1 =\bar{\bd{B}}$, the equation is equivalent to
\begin{equation}
\sum_{u=1}^5 \sum_{v=1}^5 (\bd{A}_u \bd{X}_v^1) \otimes (\J{u}\J{v}) = \sum_{j=1}^5 \bd{B}_j \otimes \J{j}
\end{equation}
which, by the product table \cref{eq-matrix-product-table}, is indeed equivalent to the linear system $\widetilde{\bd{A}}\widetilde{\bd{X}}^1 = \widetilde{\bd{B}}$, where a \emph{stacked} form of $\bar{\bd{X}}^1$ is used
\begin{equation}
\widetilde{\bd{A}} = \begin{bmatrix}
\bd{A}_1 & \bd{0} & \bd{0} & \bd{0} & \bd{0} \\
\bd{0} & \bd{A}_2 & -\kr^2\bd{A}_4 & \bd{0} & \bd{0}\\
\bd{0} & \bd{A}_3 & \bd{A}_1 - \kr^2 \bd{A}_5 & \bd{0} & \bd{0} \\
\bd{A}_4 & \bd{0} & \bd{0} & \bd{A}_2 & -\kr^2\bd{A}_4 \\
\bd{A}_5 & \bd{0} & \bd{0} & \bd{A}_3 & \bd{A}_1 - \kr^2 \bd{A}_5
\end{bmatrix},\quad
\widetilde{\bd{X}}^1 = \begin{bmatrix}
\bd{X}_1^1 \\ \vdots \\ \bd{X}_5^1
\end{bmatrix},\quad
\widetilde{\bd{B}} = \begin{bmatrix}
\bd{B}_1 \\ \vdots \\ \bd{B}_5
\end{bmatrix}.
\end{equation}
Let $a \le \min(5p,5r)$ be the rank of $\widetilde{\bd{A}} \in \mathbb{F}_0^{5p \times 5r}$.
The diagonalization of $\widetilde{\bd{A}}$ implies the existence of full-rank matrices $\bd{S} \in \mathbb{F}_0^{5p \times 5p}$ and $\bd{T} \in \mathbb{F}_0^{5r \times 5r}$ such that
\begin{equation}
\widetilde{\bd{A}} = \bd{S}\cdot \begin{bmatrix}
\bd{I}_a & \bd{0}_{a \times (5r-a)} \\
\bd{0}_{(5p-a)\times a} & \bd{0}_{(5p-a)\times(5r-a)}
\end{bmatrix}\cdot \bd{T},
\end{equation}
so the entries lower than the $a$-th row of $\bd{S}^{-1}\widetilde{\bd{B}}$ are all zero, and
\begin{equation}
\widetilde{\bd{X}}^0 = \bd{T}^{-1}\cdot \begin{bmatrix}
\begin{bmatrix}
\bd{I}_a & \bd{0}_{a\times(5p-a)}
\end{bmatrix} \cdot \left(\bd{S}^{-1}\widetilde{\bd{B}}\right) \\
\bd{0}_{(5r-a)\times q}
\end{bmatrix}\in\mathbb{F}_0^{5r \times q}
\end{equation}
will satisfy $\widetilde{\bd{A}}\widetilde{\bd{X}}^0=\widetilde{\bd{B}}$.
Write $\widetilde{\bd{X}}^0$ in the stacked form
\begin{equation}
\tilde{\bd{X}}^0 = \begin{bmatrix}
\bd{X}_1^0 \\ \vdots \\ \bd{X}_5^0
\end{bmatrix}
\end{equation}
where each $\bd{X}_j^0 \in \mathbb{F}_0^{r\times q}$, the matrix
\begin{equation}
\bar{\bd{X}}_0 = \sum_{j=1}^{5}\bd{X}_j^0 \otimes \J{j} \in \mathfrak{R}_{r \times q}(\mathbb{F}_0)
\end{equation}
is as desired.
\end{proof}

The coming discussion in later sections on the LMDG of the Maxwell's equations and the elastic wave equation are generally presented in the following pattern: the restricting equations of the problem, including linear equations derived from the interface conditions and the radiation conditions, are re-formatted using the matrix basis $\J{1}, \cdots, \J{9}$.
The solution filtering theorem then works on the linear system of the tensors in the layers, so that filtered solutions are proven available using the matrix basis formulation of $\mathfrak{R}^0$.
Finally, the formulation helps simplify the restricting equations, leaving only the basis coefficients to be solved.

\section{Application to the Maxwell's equations in layered media}\label{sect-3}
In this section we first give a brief introduction about the dyadic Green's functions of the time-harmonic Maxwell's equations in the free space, then we will discuss the Green's functions in layered media, and the simplification using the matrix basis formulation \cref{eq-R0}.

\subsection{The dyadic Green's functions of the Maxwell's equations in the free space}
Suppose in the free space the medium has constant permittivity $\varepsilon$ and constant permeability $\mu$.
Assuming the time dependence is harmonic, i.e. in terms of $\exp(\i \omega t)$, the source-free Maxwell's equations in the free space is simplified in the $\omega$-domain of the Fourier transform as the equations of the electric displacement flux $\vec{D}(\bd{r})$, the electric field $\vec{E}(\bd{r})$, the magnetic flux density $\vec{B}(\bd{r})$ and the magnetic field $\vec{H}(\bd{r})$, namely,
\begin{align}\label{eq-Maxwell-free-space}
\begin{split}
\vec{D} &= \varepsilon \vec{E} \\
\vec{B} &= \mu \vec{H} \\
\nabla \times \vec{E} &= -\i \omega \mu \vec{H}\\
\nabla \times \vec{H} &= \i\omega \varepsilon\vec{E} \\
\nabla \cdot \vec{D} &= 0 \\
\nabla \cdot \vec{B} &= 0
\end{split}
\end{align}
and the wave number is given by
\begin{equation}
k = \sqrt{\omega^2\varepsilon\mu}.
\end{equation}
When dealing with these equations, the Lorenz gauge condition is often introduced, which allows us to use a vector potential $\vec{A}(\bd{r})$ to represent the electric field $\vec{E}$ and the magnetic field $\vec{H}$ as
\begin{equation}
\vec{E} = -\i \omega \left(\bd{I} + \frac{\nabla\nabla}{k^2}\right)\vec{A},\quad \vec{H} = \frac{1}{\mu} \nabla \times \vec{A},
\end{equation}
and the flux vectors $\vec{D}$ and $\vec{B}$ are replaced using $\varepsilon\vec{E}$ and $\mu\vec{H}$, respectively.
The vector potential satisfies the Helmholtz equation
\begin{equation}
\nabla^2 \vec{A} + k^2 \vec{A} = \vec{0}.
\end{equation}
The choice of the vector potential is \emph{not} unique.
Indeed, for any function $\phi \in C^2(\mathbb{R}^3)$ satisfying the Helmholtz equation $\nabla^2 \phi + k^2 \phi = 0$, $\vec{A} + \nabla \phi$ can be used to replace $\vec{A}$ in the above identities.

The dyadic Green's functions for the free space Maxwell's equations are defined using a $3 \times 3$ potential tensor $\bd{G}_A(\bd{r};\bd{r}')$ such that the electric field dyadic Green's function $\bd{G}_E(\bd{r};\bd{r}')$ and the magnetic field dyadic Green's function $\bd{G}_H(\bd{r};\bd{r}')$ are represented by
\begin{equation}\label{eq-GEGH-GA}
\bd{G}_{E} = -\i \omega \left(\bd{I} + \frac{\nabla\nabla}{k^2}\right)\bd{G}_{A},\quad \bd{G}_{H} = \frac{1}{\mu} \nabla \times \bd{G}_{A}.
\end{equation}
The potential tensor satisfies the Helmholtz equation
\begin{equation}\label{eq-GA-eqn}
\nabla^2 \bd{G}_A + k^2 \bd{G}_A = \frac{1}{\i \omega} \delta(\bd{r}-\bd{r}')\bd{I}
\end{equation}
where $\bd{I}$ is the $3 \times 3$ identity matrix.
In addition, the dyadic Green's functions must satisfy the Sommerfeld radiation condition
\begin{equation}\label{eq-radiation-condition-GEGH-freespace}
\lim_{r \to \infty}r\left(\frac{\partial}{\partial r} - \i k\right)\bd{G}(\bd{r};\bd{r}')=\bd{0}
\end{equation}
for $\bd{G} = \bd{G}_E$ and $\bd{G} = \bd{G}_H$, here $r = |\bd{r}|$.

For the same reason, the tensor potential is not unique.
A commonly used solution to \cref{eq-GA-eqn} is given by
\begin{equation}\label{eq-Gaf-spatial}
\bd{G}_A^{\mathrm{f}}(\bd{r};\bd{r}') = -\frac{1}{\i\omega}\frac{e^{\i k |\bd{r}-\bd{r}'|}}{4\pi|\bd{r}-\bd{r}'|}\bd{I} = -\frac{1}{\i \omega} g^{\mathrm{f}}(\bd{r};\bd{r}')\bd{I},
\end{equation}
where $g^{\mathrm{f}}(\bd{r};\bd{r}')$ is the free space Green's function of the Helmholtz equation.
$\bd{G}_A^{\mathrm{f}}$ also satisfies the Sommerfeld radiation condition.

\subsection{The dyadic Green's functions of the Maxwell's equations in layered media}\label{sect-Green-layer}
Now suppose the space is horizontally stratified as layers $0,\cdots,L$ arranged from top to bottom, separated by planes $z=d_0, \cdots, z=d_{L-1}$ where $d_0 > \cdots > d_{L-1}$, and each layer is homogeneous with constant permittivity $\varepsilon_j$ and constant  permeability $\mu_j$, $j=0,\cdots,L$, respectively.
We append the layer index to the end of the subscript of any layer-dependent variable or function to represent its value specified in that layer, e.g. the wave number is then
\begin{equation}
k_j = \sqrt{\omega^2 \varepsilon_j \mu_j},\quad j=0,\cdots,L.
\end{equation}
in layer $j$.
For simplicity, the layer index is sometimes omitted, and one may assume the variable is a piecewise function of $z$.
These subscript rules are applied to the rest of this paper, including the later section about the elastic wave equations.

\subsubsection{The equations in the spatial domain}
The time-harmonic Maxwell's equations in the interior of each layer has the same form as in \cref{eq-Maxwell-free-space}, while the following interface conditions must be satisfied \cite{cai2013em}: between the adjacent layers,
\begin{equation}
\lbr \bd{n}\times \vec{E} \rbr =\vec{0},\quad \lbr \bd{n}\cdot \vec{D} \rbr =0, \quad \lbr \bd{n}\times \vec{H} \rbr =\vec{0},\quad \lbr \bd{n}\cdot \vec{B} \rbr =0.
\end{equation}
$\lbr \cdot \rbr$ represents the jump of the value at the interface, i.e. across the interface $z=d$,
\begin{equation}
\lbr f \rbr = \lim_{z\to d^+}f - \lim_{z\to d^-}f.
\end{equation}

The dyadic Green's functions are again given using the tensor potential $\bd{G}_A$ as
\begin{equation}\label{eq-GE-GH-repre}
\bd{G}_{E} = -\i \omega \left(\bd{I} + \frac{\nabla\nabla}{k^2}\right)\bd{G}_{A},\quad \bd{G}_{H} = \frac{1}{\mu} \nabla \times \bd{G}_{A}.
\end{equation}
The tensor potential $\bd{G}_A$ satisfies the Helmholtz equation
\begin{equation}\label{eq-GA-eqn-layer}
\nabla^2 \bd{G}_A + k^2 \bd{G}_A = \frac{1}{\i \omega} \delta(\bd{r}-\bd{r}')\bd{I},
\end{equation}
while the interface conditions are
\begin{equation}\label{eq-if-cond}
\lbr \bd{n}\times \bd{G}_E \rbr =\bd{0},\quad \lbr \varepsilon\bd{n}\cdot \bd{G}_E \rbr =\vec{0}^T, \quad \lbr \bd{n}\times \bd{G}_H \rbr =\bd{0},\quad \lbr \mu\bd{n}\cdot \bd{G}_H \rbr =\vec{0}^T.
\end{equation}
In horizontally layered media $\bd{n} = \bd{e}_3 = \begin{bmatrix}
0 & 0 & 1
\end{bmatrix}^T$.
In addition, the Green's functions must satisfy the Sommerfeld radiation conditions \cref{eq-radiation-condition-GEGH-freespace}.

\subsubsection{The equations in the frequency domain}
Take the the 2-D Fourier transform \cref{eq-2DFT} from $(x-x',y-y')$ to $(k_x, k_y)$.
Suppose $\bd{r} = (x,y,z)$ locates in layer $t$, and $\bd{r}'=(x',y',z')$ locates in layer $j$.
The layer index is default to the target layer $t$ when not specified.
Notations of the polar coordinate pair $(\kr, \alpha)$ are kept.

We begin with the separation of the $z$ variable from the tensor potential $\widehat{\bd{G}}_A$, which will lead to the reaction field decomposition.

For the gradient alternative in the frequency domain, define the notation $\widehat\nabla$ by
\begin{equation}
\widehat\nabla = \begin{bmatrix}
\i k_x \\ \i k_y \\ \partial_z
\end{bmatrix}
\end{equation}
when a function of $z$ follows it.
The $\widehat\nabla\widehat\nabla$, $\widehat\nabla^2$ now refer to $\widehat\nabla\widehat\nabla^T$ and $\widehat\nabla^T\widehat\nabla$, respectively.

Recall that the right-hand side of the Helmholtz equation \cref{eq-GA-eqn-layer} is nontrivial if and only if $\bd{r}'$ is in the same layer as $\bd{r}$, i.e. $j=t$, define
\begin{equation}\label{eq-GAf}
\widehat{\bd{G}}_A^\mathrm{r}(\bd{r};\bd{r}') = \widehat{\bd{G}}_A(\bd{r};\bd{r}') - \delta_{j,t}\widehat{\bd{G}}_A^\mathrm{f}(\bd{r};\bd{r}'),
\end{equation}
where $\delta_{j,t}$ is the Kronecker delta function.
The complementary part $\widehat{\bd{G}}_A^\mathrm{r}$ is called the reaction field, and satisfies the \emph{homogeneous} Helmholtz equation
\begin{equation}\label{eq-Helmholtz-freq}
\widehat\nabla^2 \widehat{\bd{G}}_A^\mathrm{r} + k^2 \widehat{\bd{G}}_A^\mathrm{r} = \bd{0},\quad\text{i.e.}\quad \partial_{zz} \widehat{\bd{G}}_A^\mathrm{r} + (k^2 - \kr^2)\widehat{\bd{G}}_A^\mathrm{r} = \bd{0}.
\end{equation}
Define
\begin{equation}
k_z = \sqrt{k^2 - \kr^2}
\end{equation}
where the square root takes nonnegative real part.
The general solutions to \cref{eq-Helmholtz-freq}, when treated as an ordinary differential equation of $z$, is given by
\begin{equation}\label{eq-GAr}
\widehat{\bd{G}}_A^\mathrm{r} = e^{\i k_z z} \widehat{\bd{G}}_A^{\mathrm{r}\uparrow} + e^{-\i k_z z} \widehat{\bd{G}}_A^{\mathrm{r}\downarrow}
\end{equation}
where $\widehat{\bd{G}}_A^{\uparrow}$ and $\widehat{\bd{G}}_A^{\downarrow}$ are piecewise constants with respect to $z$, namely,
\begin{equation}
\widehat{\bd{G}}_A^{\mathrm{r}\uparrow} = \widehat{\bd{G}}_{A,t}^{\mathrm{r}\uparrow}, \quad \widehat{\bd{G}}_A^{\mathrm{r}\downarrow} = \widehat{\bd{G}}_{A,t}^{\mathrm{r}\downarrow}
\end{equation}
when in layer $t$.
Indeed, we can also write $\widehat{\bd{G}}_A^\mathrm{f}$ in the frequency domain in a similar form
\begin{align}\label{eq-GA-free-freq}
\widehat{\bd{G}}_A^\mathrm{f} = \frac{-1}{2 \omega k_{z,j}} e^{\i k_{z,j} |z-z'|}\bd{I} = \frac{-1}{2 \omega k_{z,j}} e^{\i k_{z,j} (z-z')} 1_{\{z>z'\}}\bd{I} + \frac{-1}{2 \omega k_{z,j}} e^{\i k_{z,j} (z'-z)} 1_{\{z<z'\}}\bd{I}
\end{align}
when $z' \ne z$.
Hence we may alternatively use the notations
\begin{align}\label{eq-GApm}
\widehat{\bd{G}}_A = e^{\i k_{z} z} \widehat{\bd{G}}_A^{\uparrow} + e^{-\i k_{z} z}\widehat{\bd{G}}_A^{\downarrow}
\end{align}
where
\begin{align}\label{eq-GA-fr}
\begin{split}
\widehat{\bd{G}}_A^{\uparrow} = \delta_{j,t} 1_{\{z>z'\}} \frac{-e^{-\i k_{z,j} z'}}{2\omega k_{z,j}}\bd{I} + \widehat{\bd{G}}_{A}^{\mathrm{r}\uparrow}, \quad
\widehat{\bd{G}}_A^{\downarrow} = \delta_{j,t} 1_{\{z<z'\}} \frac{-e^{\i k_{z,j} z'}}{2\omega k_{z,j}}\bd{I} + \widehat{\bd{G}}_{A}^{\mathrm{r}\downarrow}
\end{split}
\end{align}
assuming $z \ne d_i$, $0\le i \le L-1$ and $z \ne z'$.

We call the separation of variable $z$ in \cref{eq-GApm} the \emph{reaction field decomposition} of $\widehat{\bd{G}}_A$, since this decomposition separates the free-space part and the reaction field part, and distinguishes the wave components by propagating upwards or downwards in the vertical direction.
$e^{\tau^\ast \i k_z z}\widehat{\bd{G}}_A^{\ast}$ are called the \emph{propagation components} of $\widehat{\bd{G}}_A$ by name, where
\begin{equation}\label{def-tau}
\tau^\uparrow = +1, \quad \tau^\downarrow = -1, \quad \ast \in \{ \uparrow, \downarrow \}.
\end{equation}
The $\tau^\ast$ notations are used for the rest of this paper.
$e^{\tau^\ast \i k_z z}\widehat{\bd{G}}_A^{\mathrm{r}\ast}$ are called the \emph{reaction components} of $\widehat{\bd{G}}_A$.
Similarly for $\widehat{\bd{G}}_E$ and $\widehat{\bd{G}}_H$ we will name the corresponding terms, later following \cref{eq-GE-GH-GA-form}.
\begin{remark}
In the previous work of the Helmholtz equations \cite{bo2019hfmm,bo2020tefmm,zhang2020}, the separation of variable $z'$ was also executed, so that each reaction component was further divided by the ``propagating direction'' of $z'$.
\end{remark}

Then, we re-format the restricting equations of the tensor potential $\widehat{\bd{G}}_A$, including the interface conditions and the radiation conditions, using the matrix basis $\J{1}, \cdots, \J{9}$.

With the $z$ variable separated in the reaction field decomposition, we can further expand the $\partial_z$ operator in $\widehat{\nabla}$.
Define
\begin{equation}
\widehat\nabla^{\pm} = \begin{bmatrix}
\i k_x & \i k_y & \pm \i k_{z}
\end{bmatrix}^T.
\end{equation}
By expanding from \cref{eq-GEGH-GA}, the Green's functions $\widehat{\bd{G}}_E$ and $\widehat{\bd{G}}_H$ can be represented as the linear combinations of reaction components of $\widehat{\bd{G}}_A$, with every coefficient matrices in $\mathfrak{R}^0$:
\begin{align}\label{eq-GE-GH-GA-form}
\begin{split}
\widehat{\bd{G}}_E =& -\i\omega \left(\bd{I}+\frac{\widehat\nabla^{+} (\widehat\nabla^+)^T}{k^2}\right)e^{\i k_{z} z}\widehat{\bd{G}}_A^\uparrow -\i\omega \left(\bd{I}+\frac{\widehat\nabla^{-} (\widehat\nabla^-)^T}{k^2}\right) e^{-\i k_{z} z} \widehat{\bd{G}}_A^\downarrow \\
=& -\i\omega\left( \J{1} + \frac{\kr^2}{k^2}\J{2}+\frac{1}{k^2}\J{5} + \frac{\i k_{z}}{k^2}\J{3} + \frac{\i k_{z}}{k^2} \J{4} \right) e^{\i k_{z} z}\widehat{\bd{G}}_A^\uparrow \\
&-\i\omega\left( \J{1} + \frac{\kr^2}{k^2}\J{2}+\frac{1}{k^2}\J{5} - \frac{\i k_{z}}{k^2}\J{3} - \frac{\i k_{z}}{k^2} \J{4} \right) e^{-\i k_{z} z}\widehat{\bd{G}}_A^\downarrow,
\end{split} \\
\begin{split}
\widehat{\bd{G}}_H =& \frac{1}{\mu}\widehat\nabla^+ \times e^{\i k_{z} z} \widehat{\bd{G}}_A^\uparrow + \frac{1}{\mu}\widehat\nabla^- \times e^{-\i k_{z} z} \widehat{\bd{G}}_A^\downarrow \\
=& \frac{1}{\mu}\left(\J{6}+\J{7} - \i k_{z} \J{9}\right) e^{\i k_{z} z} \widehat{\bd{G}}_A^\uparrow + \frac{1}{\mu}\left(\J{6}+\J{7} + \i k_{z} \J{9}\right) e^{-\i k_{z} z} \widehat{\bd{G}}_A^\downarrow.
\end{split}
\end{align}

The $\bd{n} \cdot$ and $\bd{n} \times$ operators in interface conditions \cref{eq-if-cond}, given $\bd{n} = \bd{e}_3$, are then converted to their equivalent matrix forms in the frequency domain, respectively, as
\begin{align}\label{eq-if-cond-J}
\lbr \J{1} \widehat{\bd{G}}_E \rbr = \bd{0}, \quad \lbr \varepsilon\J{2} \widehat{\bd{G}}_E \rbr = \bd{0}, \quad \lbr \J{9} \widehat{\bd{G}}_H \rbr = \bd{0}, \quad \lbr \mu \J{7} \widehat{\bd{G}}_H \rbr = \bd{0}.
\end{align}
For instance, the first equation in \cref{eq-if-cond} is equivalent to the continuity of the first two rows of $\widehat{\bd{G}}_E$ across the interface, corresponding to the first equation of the above.
In details, in each pair of brackets,
\begin{align}
\begin{split}
\J{1} \widehat{\bd{G}}_E &= -\i\omega\left( \J{1} +\frac{1}{k^2}\J{5} + \frac{\i k_z}{k^2}\J{3} \right)e^{\i k_z z}\widehat{\bd{G}}_A^\uparrow -\i\omega\left( \J{1} +\frac{1}{k^2}\J{5} - \frac{\i k_z}{k^2}\J{3} \right)e^{-\i k_z z}\widehat{\bd{G}}_A^\downarrow, \\
\varepsilon\J{2} \widehat{\bd{G}}_E &= -\i\omega\varepsilon\left(\frac{\kr^2}{k^2}\J{2} + \frac{\i k_z}{k^2} \J{4} \right) e^{\i k_z z} \widehat{\bd{G}}_A^\uparrow -\i\omega\varepsilon\left(\frac{\kr^2}{k^2}\J{2} - \frac{\i k_z}{k^2} \J{4} \right) e^{-\i k_z z} \widehat{\bd{G}}_A^\downarrow, \\
\J{9}\widehat{\bd{G}}_H &= -\frac{1}{\mu} \left( \J{3} - \i k_z \J{1} \right) e^{\i k_z z} \widehat{\bd{G}}_A^\uparrow -\frac{1}{\mu} \left( \J{3} + \i k_z \J{1} \right) e^{-\i k_z z} \widehat{\bd{G}}_A^\downarrow \\
\mu\J{7} \widehat{\bd{G}}_H &= \left(\kr^2\J{1}+\J{5}\right)e^{\i k_z z}\widehat{\bd{G}}_A^\uparrow + \left(\kr^2\J{1}+\J{5}\right)e^{-\i k_z z}\widehat{\bd{G}}_A^\downarrow.
\end{split}
\end{align}
The particular choice of $\J{9}$ and $\J{7}$ in \cref{eq-if-cond-J} ensures that we can simply use $\J{1}, \cdots, \J{5}$ as factors in the above equations.
When expanded across any interface $z = d_l$, each of them is a linear equation of $\widehat{\bd{G}}_{A,l}^{\ast}$ and $\widehat{\bd{G}}_{A,l+1}^{\ast}$ with coefficients in $\mathfrak{R}^0$.
Take the last one as an example.
Due to \cref{eq-GA-fr} from the reaction field decomposition, the equation $\lbr \mu \J{7} \widehat{\bd{G}}_H \rbr = \bd{0}$ is expanded with the $z \to d_l^{+}$ side as
\begin{align}
\begin{split}
&{}\left(\kr^2\J{1}+\J{5}\right)e^{\i k_{z,l} z}\widehat{\bd{G}}_{A,l}^\uparrow + \left(\kr^2\J{1}+\J{5}\right)e^{-\i k_{z,l} z}\widehat{\bd{G}}_{A,l}^\downarrow \\
={}&{}(\kr^2 \J{1} + \J{5}) \left(\widehat{\bd{G}}_{A,l}^{\mathrm{r}\uparrow} + \delta_{j,l} 1_{\{d_l > z'\}} \frac{-e^{-\i k_{z,j} z'}}{2\omega k_{z,j}}(\J{1} + \J{2}) \right)e^{\i k_{z,l}z} \\
{}+{}&{}(\kr^2 \J{1} + \J{5}) \left(\widehat{\bd{G}}_{A,l}^{\mathrm{r}\downarrow} + \delta_{j,l} 1_{\{d_l < z'\}} \frac{-e^{\i k_{z,j} z'}}{2\omega k_{z,j}}(\J{1} + \J{2}) \right)e^{-\i k_{z,l}z}
\end{split}
\end{align}
and will be able to be written using elements of $\mathfrak{R}^0$ as coefficients.
The same result applies to the $z \to d_l^{-}$ side in layer $l+1$.
So is the continuity equation itself at $z=d_l$.

For the Sommerfeld radiation conditions, it is sufficient to describe them in the frequency domain as the decay conditions of $\widehat{\bd{G}}_E$ and $\widehat{\bd{G}}_H$ as $z \to \pm \infty$, so that waves never \emph{come} from $z = \pm \infty$.
Such conditions are sufficient to uniquely determine the Green's function, so we don't bother to raise more complicated statements.
In the top layer, by \cref{eq-GE-GH-GA-form}, the downwards propagation components of $\widehat{\bd{G}}_E$ and $\widehat{\bd{G}}_H$ must be zero, since its asymptotic behavior is determined by the $e^{-\i k_{z,0} z}$ factor, so
\begin{align}\label{eq-radiation-GE}
-\i\omega\left( \J{1} + \frac{\kr^2}{k_0^2}\J{2}+\frac{1}{k_0^2}\J{5} - \frac{\i k_{z,0}}{k_0^2}\J{3} - \frac{\i k_{z,0}}{k_0^2} \J{4}\right)\widehat{\bd{G}}_{A,0}^\downarrow &= \bd{0}, \\ \label{eq-radiation-GH}
\frac{1}{\mu_0}\left(\J{6}+\J{7} + \i k_{z,0} \J{9}\right) \widehat{\bd{G}}_{A,0}^\downarrow &= \bd{0}
\end{align}
where $\widehat{\bd{G}}_{A,0}^{\downarrow} = \widehat{\bd{G}}_{A,0}^{\mathrm{r}\downarrow}$ ever since $z > z'$.
When equations \cref{eq-radiation-GE} and \cref{eq-radiation-GH} are treated as linear equations of $\widehat{\bd{G}}_{A,0}^{\mathrm{r}\downarrow}$, both coefficient matrices have rank $2$, and lead to the same general solution
\begin{equation}
\widehat{\bd{G}}_{A,0}^{\mathrm{r}\downarrow} = \widehat{\nabla}_0^- \cdot \bd{v}^T
\end{equation}
for arbitrary $3 \times 1$ vector $\bd{v}$.
Hence we can discard \cref{eq-radiation-GH} and keep only \cref{eq-radiation-GE} where an element of $\mathfrak{R}^0$ is multiplied by $\widehat{\bd{G}}_{A,0}^{\downarrow}$.
Similarly, as $z \to -\infty$ we get another equation in the bottom layer
\begin{align}\label{eq-radiation-GE-2}
-\i\omega\left( \J{1} + \frac{\kr^2}{k_L^2}\J{2}+\frac{1}{k_L^2}\J{5} + \frac{\i k_{z,L}}{k_L^2}\J{3} + \frac{\i k_{z,L}}{k_L^2} \J{4}\right)\widehat{\bd{G}}_{A,L}^\uparrow = \bd{0}
\end{align}
where $\widehat{\bd{G}}_{A,L}^{\uparrow} = \widehat{\bd{G}}_{A,L}^{\mathrm{r}\uparrow}$ ever since $z < z'$.

The interface conditions \cref{eq-if-cond-J} and the radiation conditions \cref{eq-radiation-GE} and \cref{eq-radiation-GE-2} in total consist a linear system of the unknown tensors $\widehat{\bd{G}}_{A,t}^{\mathrm{r}\ast}$ in the reaction field for $0 \le t \le L$ and $\ast \in \{ \uparrow, \downarrow \}$, with coefficients in $\mathfrak{R}^0$.
They are in general sufficient to uniquely determine all the $\widehat{\bd{G}}_{A,t}^{\mathrm{r}\ast}$ terms.
By \Cref{thm-solution-filtering}, there exists a solution to this linear system with each unit of the block in $\mathfrak{R}^0$, i.e. satisfying each $\widehat{\bd{G}}_{A}^{\mathrm{r}\uparrow}, \widehat{\bd{G}}_{A}^{\mathrm{r}\downarrow} \in \mathfrak{R}^0$ piecewisely in each layer.
It also follows from the reaction field decomposition \cref{eq-GApm} that
\begin{equation}
\widehat{\bd{G}}_A \in \mathfrak{R}^0
\end{equation}
has the matrix basis formulation.

\subsubsection{Further simplification of the formulation}
With the matrix basis formulation we are able to further simplify the interface equations \cref{eq-if-cond-J} and the radiation equations \cref{eq-radiation-GE} and \cref{eq-radiation-GE-2}.
Suppose $\bd{G}_{A}^{\mathrm{r}\ast} \in \mathfrak{R}^0$ has the basis expansion
\begin{equation}\label{eq-GA-matrix-basis}
\bd{G}_{A}^{\mathrm{r}\ast} = \sum_{l=1}^{5} a_{l}^{\mathrm{r}\ast} \J{l}, \quad \ast \in \{ \uparrow, \downarrow\}.
\end{equation}
Corresponding to \cref{eq-GApm}, define
\begin{equation}\label{eq-a-decomp}
a_l = \delta_{j,t} a_l^{\mathrm{f}} + e^{\i k_{z}} a_{l}^{\mathrm{r}\uparrow} + e^{-\i k_{z}} a_{l}^{\mathrm{r}\downarrow},
\end{equation}
where $a_{l}^\mathrm{f}$ are the matrix basis coefficients of the free space potential tensor \cref{eq-GA-free-freq}
\begin{equation}
\sum_{l=1}^{5} a_{l}^{\mathrm{f}} \J{l} = \widehat{\bd{G}}_A^{\mathrm{f}} = \frac{-1}{2\omega k_{z,j}} e^{\i k_{z,j}|z-z'|} (\J{1} + \J{2}),
\end{equation}
then we have derived the matrix basis formulation for $\widehat{\bd{G}}_A = \sum_{l=1}^{5} a_{l} \J{l}$ with the reaction field decomposition of each $a_l$.
It is straightforward that each coefficient $a_{l}$ satisfies the Helmholtz equation
\begin{equation}\label{eq-a-Helmholtz}
\partial_{zz} a_{l} + k_{z}^2 a_{l} = 0
\end{equation}
piecewisely in each layer, provided $z \ne z'$.

However, the potential tensor $\widehat{\bd{G}}_A$ is still \emph{not} uniquely determined if assumed only having a matrix basis representation.
For instance, for any functions $f_1, f_2 \in C^2(\mathbb{R}^3)$ satisfying the Helmholtz equation $\nabla^2 f_j + k^2 f_j = 0$, $j=1,2$, the potential tensor $\widehat{\bd{G}}_A + \partial_z \widehat{f}_1 \J{2} + \widehat{f}_1 \J{3} + \partial_z \widehat{f}_2 \J{4} + \widehat{f}_2 \J{5}$ can be used to replace $\widehat{\bd{G}}_A$.
To eliminate the degrees of freedom in the coefficients, define the functions $b_1$, $b_2$ and $b_3$ by linear transforms of $a_{l}$:
\begin{align}\label{eq-b}
\begin{split}
b_1 &= a_1, \\
b_2 &= \frac{1}{\mu}\left( a_2 - \partial_z a_3 \right), \\
b_3 &= \frac{1}{\mu}\left( \partial_z a_1 + \kr^2 a_4 - \kr^2 \partial_z a_5 \right),
\end{split}
\end{align}
so that $\widehat{\bd{G}}_E$ and $\widehat{\bd{G}}_H$ in \cref{eq-GE-GH-repre} can be represented by
\begin{align}\label{eq-GE-GH-a}
\begin{split}
\widehat{\bd{G}}_E &= -\frac{\i \omega}{k^2} \left( k^2 b_1 \J{1} + \mu \kr^2 b_2 \J{2} + \mu \partial_z b_2 \J{3} + \mu b_3 \J{4} + \left(\frac{k^2}{\kr^2}b_1 + \frac{\mu}{\kr^2} \partial_z b_3\right)\J{5} \right), \\
\widehat{\bd{G}}_H &= \frac{1}{\mu}\left( b_1 \J{6} + \mu b_2 \J{7} + \left(\frac{1}{\kr^2}\partial_z b_1 - \frac{\mu}{\kr^2} b_3\right)\J{8} - \partial_z b_1 \J{9} \right).
\end{split}
\end{align}
Each $b_{l}$ has the reaction component decomposition corresponding to \cref{eq-a-decomp}
\begin{equation}\label{eq-b-decomp}
b_{l}=b_{l}(\kr, z, z') = \delta_{j,t} b_{l}^{\mathrm{f}}(\kr, z, z') + e^{\i k_{z}z} b_{l}^{\mathrm{r}\uparrow}(\kr, z') + e^{-\i k_{z}z} b_{l}^{\mathrm{r}\downarrow}(\kr, z'),
\end{equation}
where
\begin{align}
\begin{split}
b_1^\mathrm{f} &= a_1^\mathrm{f}, \\
b_2^{\mathrm{f}} &= \frac{1}{\mu}\left( a_2^\mathrm{f} - \partial_z a_3^\mathrm{f} \right), \\
b_3^{\mathrm{f}} &= \frac{1}{\mu}\left( \partial_z a_1^{\mathrm{f}} + \kr^2 a_4^{\mathrm{f}} - \kr^2 \partial_z a_5^{\mathrm{f}} \right), \\
b_1^{\mathrm{r}\ast} &= a_1^{\mathrm{r}\ast}, \\
b_2^{\mathrm{r}\ast} &= \frac{1}{\mu}\left( a_2^{\mathrm{r}\ast} - \tau^\ast \i k_{z} a_3^{\mathrm{r}\ast} \right), \\
b_3^{\mathrm{r}\ast} &= \frac{1}{\mu}\left( \tau^\ast \i k_{z} a_1^{\mathrm{r}\ast} + \kr^2 a_4^{\mathrm{r}\ast} - \kr^2 \tau^\ast \i k_{z} a_5^{\mathrm{r}\ast} \right),
\end{split}
\end{align}
due to \cref{eq-a-decomp}, here $\ast \in \{ \uparrow, \downarrow\}$, $\tau^\uparrow = 1$, $\tau^\downarrow = -1$.
Specifically, since we have chosen $\bd{G}_A = \bd{G}_A^{\mathrm{f}} = -g^{\mathrm{f}}/(\i \omega)\bd{I}$, it's clear that
\begin{equation}
a_1^\mathrm{f} = a_2^\mathrm{f} = -\frac{1}{\i\omega} \widehat g^\mathrm{f}, \quad a_3^\mathrm{f} = a_4^\mathrm{f} = a_5^\mathrm{f} = 0,
\end{equation}
where $ \widehat g^\mathrm{f} = \i e^{\i k_{z,j} |z-z'|}/(2 k_{z,j})$.
Therefore
\begin{align}
\begin{split}
b_1^\mathrm{f} &= -\frac{1}{\i \omega} \widehat{g}^\mathrm{f}, \\
b_2^\mathrm{f} &= -\frac{1}{\i \omega} \frac{1}{\mu_j} \widehat{g}^\mathrm{f}, \\
b_3^\mathrm{f} &= -\frac{1}{\i \omega} \frac{1}{\mu_j} \partial_z \widehat{g}^\mathrm{f}
= - \partial_{z'} b_2^\mathrm{f}.
\end{split}
\end{align}
Each $b_l$ also satisfies the Helmholtz equation
\begin{equation}\label{eq-b-Helmholtz}
\partial_{zz} b_j + k_z^2 b_j = 0
\end{equation}
piecewisely in each layer provided $z' \ne z$.

For solving $b_1, b_2$ and $b_3$ we take a review of the interface equations and the radiation equations.
One can easily verify the interface equations \cref{eq-if-cond}, which were reinterpreted in the frequency domain as in \cref{eq-if-cond-J}, are equivalent to the following by comparing the matrix basis coefficients.
For example, from
\begin{equation*}
\bd{J}_1 \cdot \widehat{\bd{G}}_E = -\i\omega \left( b_1 \J{1} + \omega^{-2}\varepsilon^{-1} \partial_z b_2 \J{3} + ( \kr^{-2} b_1 + \omega^{-2}\varepsilon^{-1} \kr^{-2} \partial_z b_3 ) \J{5} \right)
\end{equation*}
the continuity equations
\begin{equation*}
\lbr -\i\omega b_1 \rbr = 0, \quad \lbr -\i \omega^{-1} \varepsilon^{-1} \partial_z b_2 \rbr = 0, \quad \lbr -\i \omega \kr^{-2} b_1  - \i \omega^{-1} \kr^{-2} \varepsilon^{-1} \rbr = 0
\end{equation*}
are revealed.
A complete list by items is given below:
\begin{align}\label{eq-if-equiv-b}
\begin{split}
\lbr \bd{n}\times \widehat{\bd{G}}_E \rbr =\bd{0} \Leftrightarrow \lbr\J{1}\cdot \widehat{\bd{G}}_E\rbr = \bd{0}
&\Leftrightarrow
\lbr b_1 \rbr =0, \left\lbr \frac{1}{\varepsilon}\partial_z b_2 \right\rbr = 0, \left\lbr\frac{1}{\varepsilon}\partial_z b_3\right\rbr = 0 ; \\
\lbr\varepsilon\bd{n}\cdot \widehat{\bd{G}}_E \rbr =\vec{0} \Leftrightarrow \lbr\J{2} \cdot \varepsilon \widehat{\bd{G}}_E\rbr = \bd{0}
&\Leftrightarrow \lbr b_2 \rbr =0, \lbr b_3 \rbr = 0 ; \\
\lbr \bd{n}\times \widehat{\bd{G}}_H \rbr =\bd{0} \Leftrightarrow \lbr \J{9}\cdot \widehat{\bd{G}}_H \rbr = \bd{0}
&\Leftrightarrow \lbr b_2 \rbr =0, \lbr b_3 \rbr = 0, \left\lbr\frac{1}{\mu}\partial_z b_1\right\rbr=0 ; \\
\lbr\mu\bd{n}\cdot \widehat{\bd{G}}_H \rbr =\vec{0} \Leftrightarrow \lbr \J{7}\cdot \mu\widehat{\bd{G}}_H \rbr =\bd{0}
&\Leftrightarrow
\lbr b_1 \rbr=0.
\end{split}
\end{align}
The radiation equations \cref{eq-radiation-GE} and \cref{eq-radiation-GE-2} are reduced to
\begin{equation}
b_{l,0}^{\mathrm{r}\downarrow} = 0, \quad b_{l,L}^{\mathrm{r}\uparrow} = 0
\end{equation}
in the top and the bottom layer, respectively, i.e. waves coming from $z = \pm \infty$ are prohibited in the reaction field decomposition.
For each $l$, the above are a total of $2L+2$ linear equations of $b_{l}^{\mathrm{r}\uparrow}$ and $b_{l}^{\mathrm{r}\downarrow}$ from $L+1$ layers.
These linear equations are solvable, from the knowledge of the acoustic wave equation in layered media:
\begin{itemize}
\item $-\i\omega b_1$ is exactly the reflection/transmission coefficient in the frequency domain of the Green's function of the Helmholtz equation in layered media, with piecewise constant material parameters $1 / \varepsilon$.
Thus we can solve $b_1$ in the frequency domain like solving the known scalar layered Helmholtz problem \cite{bo2020tefmm}.
\item Similarly, $-\i\omega \mu_{j} b_2$ is exactly the one with piecewise constant parameters $1/\mu$.
\item The linear system regarding $b_3^{\mathrm{r}\ast}$ has exactly the same coefficients as $b_2^{\mathrm{r} \ast}$ for the unknowns, so it's solvable since $b_2^{\mathrm{r}\ast}$ are uniquely determined by the physical problem.
Moreover,
\begin{equation}
-\partial_{z'} b_2 = -\partial_{z'} \delta_{j,t} b_2^{\mathrm{f}} - e^{\i k_{z} z} \partial_{z'} b_2^{\mathrm{f} \uparrow} - e^{-\i k_{z} z} \partial_{z'} b_2^{\mathrm{f} \downarrow}
\end{equation}
satisfies every equation that $b_3$ should satisfy, so by uniqueness,
\begin{equation}
b_3 = -\partial_{z'} b_2, \text{ i.e. } b_3^{\mathrm{r}\ast} = -\partial_{z'} b_2^{\mathrm{r}\ast}.
\end{equation}
\end{itemize}
The $b_1$ and $b_2$ functions are corresponding to the TE mode component and the TM mode component in the $E_z$-$H_z$ formulation \cite{kong1972, chew2006matrix}, respectively.

\begin{remark}
To characterize $\widehat{\bd{G}}_E$ and $\widehat{\bd{G}}_H$ we don't need the intermediate, undetermined tensor potential $\widehat{\bd{G}}_A$ anymore.
This paper derived the above formulation via $\widehat{\bd{G}}_A$ because some previous work did need its formulation, such as in the integral equation applications in \cite{michalski1990}.
\end{remark}

\begin{remark}
If the interface conditions are not exactly proposed like the above, e.g. proposed on the boundary of a half-space problem where only two of the interface equations \cref{eq-if-cond} hold, the result of the matrix basis formulation still holds with the same derivation.
\end{remark}

\begin{remark}[modes of the system]
A mode of the layered media is an eigenstate without stimulation from any given source, i.e. the nontrivial solution of $\widehat{\bd{G}}_A^{\mathrm{r}\ast}$ satisfying the above interface equations and radiation equations for certain values of $\kr$, with each $\widehat{\bd{G}}_A^{\mathrm{f}}$ replaced by $0$.
It is corresponding to a pole in the frequency domain \cite{zhang2020}.
In such situation we can still derive the simplified formulation using terms $b_1, b_2$ and $b_3$, but $b_3$ plays an independent role and is not anymore tied with $b_2$.
\end{remark}

\subsubsection{The transverse potential and the Sommerfeld potential}
Here we take a quick review on the transverse potential and the Sommerfeld potential formulations and show how to reach them from the matrix basis formulation.
Both formulations restrict certain 5 entries of the $3 \times 3 $ tensor $\widehat{\bd{G}}_A$ to be nonzero, which uniquely determines the tensor potential.
Here we claim the potential tensors in these formulations have the matrix basis representation, and can be derived using $b_1$ and $b_2$.
Due to the uniqueness of $b_1$ and $b_2$, it suffices to explicitly construct them.

The transverse potential takes the form
\begin{equation}
\widehat{\bd{G}}_A^{\mathrm{t}} = \begin{bmatrix}
\times & \times & \\
\times & \times & \\
& & \times
\end{bmatrix},
\end{equation}
where each $\times$ marks a nonzero entry.
We claim $\widehat{\bd{G}}_A^{\mathrm{t}} = a_1 \J{1} + a_2 \J{2} + a_5 \J{5}$.
By \cref{eq-b},
\begin{equation}
b_1 = a_1, \quad b_2 = \frac{1}{\mu} a_2, \quad b_3 = -\partial_{z'} b_2 = \frac{1}{\mu}(\partial_z a_1 - \kr^2 \partial_z a_5).
\end{equation}
Since $a_l$, $b_l$ satisfy the Helmholtz equation \cref{eq-a-Helmholtz} and \cref{eq-b-Helmholtz}, respectively, we have
\begin{equation}
a_1 = b_1, \quad a_2 = \mu b_2, \quad a_5 = b_1 - \frac{\mu \partial_{z} \partial_{z'} b_2}{\kr^2 k_{z}^2}.
\end{equation}

The Sommerfeld potential takes the form
\begin{equation}
\bd{G}_A^{\mathrm{S}} = \begin{bmatrix}
\times & & \\
& \times & \\
\times & \times & \times
\end{bmatrix}
\end{equation}
and we claim $\bd{G}_A^{\mathrm{S}} = a_1 \J{1} + a_2 \J{2} + a_4 \J{4}$.
Again by \cref{eq-b},
\begin{equation}
b_1 = a_1, \quad b_2 = \frac{1}{\mu} a_2, \quad b_3 = -\partial_{z'} b_2 = \frac{1}{\mu}(\partial_z a_1 + \kr^2 a_4),
\end{equation}
so
\begin{equation}
a_1 = b_1, \quad a_2 = \mu b_2, \quad a_4 = -\frac{\mu \partial_{z'} b_2 + \partial_{z} b_1}{\kr^2}.
\end{equation}

\begin{remark}
In the transverse potential $\widehat{\bd{G}}_A^{\mathrm{t}}$, although the coefficient $a_5$ has a $\kr^2$ factor in the denominator, there's no singularity in the integrand of $\widehat{\bd{G}}_A^{\mathrm{t}}$ at $\kr = 0$ since they can be cancelled out with the entries of $\J{5}$.
The same does happen to the Sommerfeld potential $\widehat{\bd{G}}_A^{\mathrm{S}}$, but it's not explicitly shown in the expression of $a_4 \J{4}$.
Numerically we should take some care if the values as $\kr \to 0$ are required.
\end{remark}

\section{Application to the elastic wave equation in layered media}\label{sect-4}
In this section we will apply the matrix based formulation to the dyadic Green's function of the elastic wave equation in layered media.
The interface conditions for various contacting media phase types will be discussed.
In the end we also give a brief discussion on the case when source is inside zero-viscosity fluid layer, with a simplified vector basis formulation.
In these problems, we suppose the displacement of the media is small, so that linearization of the elastic wave equations is in general applicable.

\subsection{The dyadic Green's function of the elastic wave equation in the free space}
Suppose the homogeneous and isotropic material occupies the entire space, with density $\rho$ and Lam\'e constants $\lambda, \mu$.
Define
\begin{equation}\label{eq-gamma}
\gamma = \lambda + 2\mu.
\end{equation}
If the material is solid, both $\lambda, \mu >0$.
If the material is (compressible) liquid or gas, we assume the viscosity is negligible (also for the rest of this paper; otherwise we should treat it in the same way as a solid layer), then $\mu = 0$, and $\gamma = \lambda$.
Note that the shear modulus $\mu$ is different from the dynamic viscosity in fluid.
Suppose the time dependence is harmonic, i.e. in terms of $\exp(\i \omega t)$.

With external force $\bd{b}(\bd{r})$, the elastic wave equation in solid is a partial differential equation of the displacement $\bd{u}(\bd{r})$
\begin{equation}\label{eq-ewq-st}
-\omega^2\rho \bd{u} = \nabla \cdot \bs{\mathcal{T}} + \bd{b}
\end{equation}
assuming $|\bd{u}| \ll 1$, where the $3 \times 3$ stress tensor $\bs{\mathcal{T}}$ is defined as
\begin{equation}\label{eq-stress-tensor}
\bs{\mathcal{T}}_{ij} = \lambda \delta_{i,j} \sum_{l=1}^{3} \frac{\partial u_l}{\partial x_l} + \mu \left( \frac{\partial u_i}{\partial x_j} + \frac{\partial u_j}{\partial x_i} \right),
\end{equation}
where $(x_1, x_2, x_3)$ is used as an alternative notation for $(x,y,z)$, and $\bd{u} = (u_1, u_2, u_3)$ \cite{chew2008}.
The equivalent form not using the stress tensor $\bs{\mathcal{T}}$ is
\begin{equation}\label{eq-ewq}
(\lambda+\mu)\nabla \nabla \cdot \bd{u} + \mu \nabla^2 \bd{u} + \omega^2 \rho \bd{u} = -\bd{b}.
\end{equation}
Given source location $\bd{r}'$, the dyadic Green's function $\bd{G}(\bd{r};\bd{r}')$ is a $3 \times 3$ tensor satisfying the equation
\begin{equation}\label{eq-ewq-G}
(\lambda+\mu)\nabla \nabla \cdot \bd{G} + \mu \nabla^2 \bd{G} + \omega^2 \rho \bd{G} = -\delta(\bd{r} - \bd{r}')\bd{I}.
\end{equation}
The solution is known as
\begin{equation}\label{eq-ewq-Gf}
\bd{G} = \bd{G}^{\mathrm{f}} (\bd{r};\bd{r}') = \frac{1}{\mu}\left(\bd{I} + \frac{\nabla\nabla}{k_s^2}\right)g_s(\bd{r};\bd{r}') - \frac{1}{\gamma}\frac{\nabla\nabla}{k_c^2}g_c(\bd{r};\bd{r}'),
\end{equation}
where
\begin{equation}
k_s = \sqrt{\frac{\omega^2\rho}{\mu}}, \quad k_c = \sqrt{\frac{\omega^2\rho}{\gamma}}
\end{equation}
are the wave numbers of the S-wave and the P-wave, respectively, and
\begin{equation}
g_s(\bd{r};\bd{r}') = \frac{e^{\i k_s |\bd{r}-\bd{r}'|}}{4\pi|\bd{r}-\bd{r}'|}, \quad g_c(\bd{r};\bd{r}') = \frac{e^{\i k_c |\bd{r}-\bd{r}'|}}{4\pi|\bd{r}-\bd{r}'|}
\end{equation}
are the free space Green's functions of the Helmholtz equation with wave numbers $k_s$ and $k_c$, respectively \cite{chew2008}.

In fluid, i.e. in liquid or gas, assuming the external force $\bd{b}$ is conservative (for the rest of this paper as well), there's only the P-wave propagating in the media, and the acoustic wave equation with respect to the displacement $\bd{u}$ is
\begin{align}\label{eq-euler-fluid}
\lambda \nabla \nabla \cdot \bd{u} + \omega^2 \rho \bd{u} &= -\bd{b}, \\
\nabla \times \bd{u} &= \bd{0},
\end{align}
where the first equation repeats \cref{eq-ewq}, and the second one is introduced because of zero viscosity.
Also consider the linearized Navier--Stokes equation
\begin{equation}\label{eq-ns-fluid}
-\omega^2 \rho \bd{u} = \rho \left(\frac{D\bd{v}}{Dt}\right)^{\wedge} = -\nabla p + \bd{b}
\end{equation}
where $\bd{v}$ is the velocity, $D/Dt$ is the material derivative and $\wedge$ represents the Fourier transform of time.
For the pressure $p$ we get
\begin{equation}
p = -\lambda \nabla \cdot \bd{u} + p_0
\end{equation}
where $p_0$ is a constant.
For the sake of convenience let $p_0 = 0$.
Take the divergence of the first equation of \cref{eq-euler-fluid},
\begin{equation}\label{eq-ewq-p}
\nabla^2 p + \frac{\omega^2 \rho}{\lambda} p = \frac{1}{\lambda}\nabla \cdot \bd{b}.
\end{equation}
The dyadic Green's function of the P-wave propagation is often proposed in terms of the pressure $p$, namely
\begin{equation}
\nabla^2 g_p + \frac{\omega^2 \rho}{\lambda} g_p = -\delta(\bd{r} - \bd{r}')
\end{equation}
with solution $g_p = g_p^{\mathrm{f}}(\bd{r};\bd{r}') = g_c(\bd{r};\bd{r}')$.
For the displacement $\bd{u}$, the corresponding Green's function $\bd{g}_{\bd{u}}^{\mathrm{f}}$ satisfies
\begin{equation}\label{eq-ewq-guf}
\gu^\mathrm{f} = \frac{1}{\omega^2 \rho} \nabla g_p^{\mathrm{f}}
\end{equation}
which straightforwardly follows \cref{eq-euler-fluid}.

\subsection{The elastic wave equation in layered media}
In layered media, again we suppose the space is horizontally stratified as layers $0,\cdots,L$ arranged from top to bottom, separated by planes $z=d_0, \cdots, z=d_{L-1}$ where $d_0 > \cdots > d_{L-1}$, and the medium in each layer is homogeneous with density $\rho_l$ and Lam\'e constants $\lambda_l$ and $\mu_l$, $l=0,\cdots,L$, respectively.
Define $\gamma = \lambda + 2\mu$, and the wave numbers
\begin{equation}
k_{s} = \sqrt{\frac{\omega^2 \rho}{\mu}},\quad k_{c} = \sqrt{\frac{\omega^2 \rho}{\gamma}}
\end{equation}
are the same as in the free space.
The notation for variables in different layers follows the convention introduced at the beginning of \Cref{sect-Green-layer}.
In fluid media, $k_s$ is not defined because $\mu = 0$ and S-wave does not propagate.

\subsubsection{The interface conditions of the elastic wave equation}
Before entering the discussion on the Green's functions, we take a detailed study on the interface conditions.
Recall the equation \cref{eq-ewq-st} in terms of the stress tensor $\bs{\mathcal{T}}$, if the external force $\bd{b}$ and the displacement $\bd{u}$ are not singular near the interface, then by divergence theorem on a flat cylinder crossing the boundary with infinitesimal thickness, we get the interface condition
\begin{equation}\label{eq-ewq-interface-st}
\left \lbr \bd{n} \cdot \bs{\mathcal{T}} \right \rbr = \vec{0}.
\end{equation}
With $\bd{n} = \bd{e}_3$ in our problem, it suffices to consider the entries
\begin{align}
\begin{split}
\mathcal{T}_{31} &= \mu \left( \frac{\partial u_3}{\partial x_1} + \frac{\partial u_1}{\partial x_3} \right), \\
\mathcal{T}_{32} &= \mu \left( \frac{\partial u_3}{\partial x_2} + \frac{\partial u_2}{\partial x_3} \right), \\
\mathcal{T}_{33} &= \gamma\frac{\partial u_3}{\partial x_3} + \lambda \left( \frac{\partial u_1}{\partial x_1} + \frac{\partial u_2}{\partial x_2} \right)
\end{split}
\end{align}
for the continuity equations.
The identity \cref{eq-ewq-interface-st} also implies some additional regular conditions to be listed below \cite{chew2008}.
\begin{itemize}
\item Across the solid-solid interface, $\partial u_3 / \partial x_3$ must be regular, so $\lbr u_3 \rbr = 0$.
Then from the continuity of $\mathcal{T}_{31}$ we further get $\lbr u_1 \rbr = 0$, and similarly $\lbr u_2 \rbr = 0$.
\item Across the solid-liquid/gas interface, we also have $\lbr u_3 \rbr = 0$.
Since in the fluid side $\mu = 0$ so that $\mathcal{T}_{31} = \mathcal{T}_{32} = 0$, no additional condition is required.
\item Across the fluid-fluid interface, the conditions on $\mathcal{T}_{31}$ and $\mathcal{T}_{32}$ are no more necessary, therefore only $\lbr \mathcal{T}_{33}\rbr = 0$ and $\lbr u_3 \rbr = 0$ are required.
Note that in this case $\mathcal{T}_{33} = \lambda \nabla \cdot \bd{u} = p$, and $u_3 = -(\omega^2 \rho)^{-1} \partial p / \partial x_3$.
\item In the half-space problem where the media has an interface against the vacuum where no acoustic wave propagates and $\lambda = \mu = 0$, the zero-traction conditions $\lbr \mathcal{T}_{3l} \rbr = 0$ are required, $l = 1,2,3$, while the displacement on the boundary is set free.
\end{itemize}
In summary, the interface equations by case are listed below:
\begin{itemize}
\item Across the solid-solid interface,
\begin{align}\label{eq-ewq-if-ss}
\lbr \mathcal{T}_{3l} \rbr = 0, \quad \lbr u_{l} \rbr = 0, \quad l = 1,2,3.
\end{align}
\item Across the solid-fluid interface,
\begin{equation}\label{eq-ewq-if-sf}
\lbr u_{3} \rbr = 0, \quad \lbr \mathcal{T}_{3l} \rbr = 0, \quad l = 1,2,3,
\end{equation}
where from the fluid side $\mathcal{T}_{31} = \mathcal{T}_{32} = 0$ automatically holds.
\item Across the fluid-fluid interface,
\begin{equation}\label{eq-ewq-if-ff}
\lbr u_{3} \rbr = 0, \quad \lbr \mathcal{T}_{33} \rbr = 0,
\end{equation}
or equivalently, in terms of the pressure $p$,
\begin{equation}\label{eq-ewq-if-ff-p}
\left\lbr \frac{1}{\rho} \frac{\partial p}{\partial z} \right \rbr = 0, \quad \lbr p \rbr = 0.
\end{equation}
\item On the solid-vacuum interface,
\begin{equation}\label{eq-ewq-if-sv}
\quad \lbr \mathcal{T}_{3l} \rbr = 0, \quad l = 1,2,3.
\end{equation}
\item On the fluid-vacuum interface,
\begin{equation}\label{eq-ewq-if-fv}
\quad \lbr \mathcal{T}_{33} \rbr = 0.
\end{equation}
\end{itemize}
To flexibly handle various interface circumstances, we separate these equations into four groups: (a) $\lbr \mathcal{T}_{33} \rbr = 0$; (b) $\lbr u_{3} \rbr = 0$; (c) $\lbr \mathcal{T}_{31} \rbr = \lbr \mathcal{T}_{32} \rbr = 0$ ; (d) $\lbr u_1 \rbr = \lbr u_2 \rbr = 0$.
For each of the above cases we simply join them in need.

When the source is in solid medium, the Green's function in terms of the displacement $\bd{G}$ is a $3 \times 3$ tensor.

\subsection{The dyadic Green's function in layered media with source in solid}\label{section-Gf-source-in-solid}
We begin with the case when all the layers are solid.
Then the case with existing fluid or vacuum layers will follow.
Suppose the source $\bd{r}' = (x',y',z')$ locates in layer $j$, and the target $\bd{r} = (x,y,z)$ is in layer $t$.
Take the 2-D Fourier transform $(x-x', y-y')$ to $(k_x, k_y)$ as defined in \cref{eq-2DFT}.
Suppose $(\kr, \alpha)$ are the polar coordinates of $(k_x, k_y)$.
Like in \cref{eq-GAf} for the Maxwell's equations, define the reaction field Green's function
\begin{equation}
\bd{G}^{\mathrm{r}} = \bd{G} - \delta_{j,t}\bd{G}^\mathrm{f},
\end{equation}
then $\bd{G}^\mathrm{r}$ satisfies the homogeneous elastic wave equation
\begin{equation}\label{eq-ewq-Gr}
(\lambda+\mu)\nabla \nabla \cdot \bd{G}^\mathrm{r} + \mu \nabla^2 \bd{G}^\mathrm{r} + \omega^2 \rho \bd{G}^\mathrm{r} = \bd{0}
\end{equation}
within each layer, which, in the frequency domain, is an ordinary differential equation
\begin{equation}\label{eq-Gr-ODE}
(\mu \J{1} + \gamma \J{2}) \partial_{zz} \widehat{\bd{G}}^\mathrm{r} + (\lambda + \mu)(\J{3} + \J{4}) \partial_{z} \widehat{\bd{G}}^\mathrm{r} + \left( \mu k_{sz}^2(\J{1} + \J{2}) + (\lambda+\mu)\J{5} \right)\widehat{\bd{G}}^\mathrm{r} = \bd{0}
\end{equation}
of $z$, here
\begin{equation}
k_{sz} = \sqrt{k_s^2 - \kr^2}, \quad k_{cz} = \sqrt{k_c^2 - \kr^2}.
\end{equation}
The above differential equation \cref{eq-Gr-ODE} has the general solution
\begin{align}\label{eq-homo-ewq-general}
\begin{split}
\widehat{\bd{G}}^\mathrm{r} = \sum_{\ast \in \{ \uparrow, \downarrow \}} \left( (- \tau^\ast \i k_{sz} \J{1} + \J{4}) e^{\tau^\ast\i k_{sz} z} + (\tau^\ast \i k_{cz} \J{2} + \J{3}) e^{\tau^\ast \i k_{cz} z} \right) \bd{X}^{\mathrm{r}\ast}
\end{split}
\end{align}
where $\tau^\uparrow = +1, \tau^\downarrow = -1$, $\bd{X}^{\mathrm{r}\ast} = \bd{X}_t^{\mathrm{r}\ast}$ is piecewise constant (with respect to the variable $z$) in each layer.
The free space part $\widehat{\bd{G}}^{\mathrm{f}}$ given by \cref{eq-ewq-Gf} takes the following form in the frequency domain:
\begin{align}\label{eq-ewq-Gf-J}
\begin{split}
\widehat{\bd{G}}^{\mathrm{f}} ={} & \frac{1}{\omega^2 \rho_j} (k_{s,j}^2 \J{1} + \kr^2 \J{2} + \i k_{sz,j} \tau \J{3} + \i k_{sz,j} \tau \J{4} + \J{5}) \widehat{g}_{s} \\
&- \frac{1}{\omega^2 \rho_j} (-k_{cz,j}^2 \J{2} + \i k_{cz,j} \tau \J{3} + \i k_{cz,j} \tau \J{4} + \J{5}) \widehat{g}_{c}
\end{split}
\end{align}
for $z \ne z'$, where $\tau$ is the sign function of $z - z'$,
\begin{equation}
\widehat{g}_s = \frac{\i e^{\i k_{sz,j}|z-z'|}}{2 k_{sz, j}}, \quad \widehat{g}_c = \frac{\i e^{\i k_{cz,j}|z-z'|}}{2 k_{cz, j}}.
\end{equation}
$\widehat{\bd{G}}^\mathrm{f}$ actually has the same form as in \cref{eq-homo-ewq-general}:
\begin{align}\label{eq-ewq-Gf-general}
\begin{split}
\widehat{\bd{G}}^\mathrm{f} = \sum_{\ast \in \{ \uparrow, \downarrow \}} \left( (-\tau^\ast \i k_{sz,j} \J{1} + \J{4}) e^{\tau^\ast\i k_{sz,j} z} + (\tau^\ast \i k_{cz,j} \J{2} + \J{3}) e^{\tau^\ast \i k_{cz,j} z} \right) \bd{X}^{\mathrm{f}\ast},
\end{split}
\end{align}
where
\begin{align}
\begin{split}
\bd{X}^{\mathrm{f} \uparrow} &= 1_{\{z>z'\}} ( x_{1}^{\mathrm{f} \uparrow} \J{1} + \cdots + x_{5}^{\mathrm{f} \uparrow} \J{5}), \\
\bd{X}^{\mathrm{f} \downarrow} &= 1_{\{z<z'\}} ( x_{1}^{\mathrm{f} \downarrow} \J{1} + \cdots + x_{5}^{\mathrm{f} \downarrow} \J{5}), \\
x_1^{\mathrm{f} \ast} &= k_{s,j}^2 D_s, \\
x_2^{\mathrm{f} \ast} &= -k_{cz,j}^2 D_c, \\
x_3^{\mathrm{f} \ast} &= \tau^\ast \i k_{sz,j} D_s, \\
x_4^{\mathrm{f} \ast} &= \tau^\ast \i k_{cz,j} D_c, \\
x_5^{\mathrm{f} \ast} &= D_s
\end{split}
\end{align}
provided $z \ne z'$, here
\begin{equation}
D_s = \frac{-1}{2\omega^2 \rho_j k_{sz,j}^2}, \quad D_c = \frac{-1}{2\omega^2 \rho_j k_{cz,j}^2}.
\end{equation}
In total, we get the reaction field decomposition of $\widehat{\bd{G}}$ by
\begin{equation}\label{eq-ewq-G-X}
\widehat{\bd{G}} = \sum_{\ast \in \{\uparrow, \downarrow\}} \left( (-\tau^\ast \i k_{sz} \J{1} + \J{4}) e^{\tau^\ast\i k_{sz} z} + (\tau^\ast \i k_{cz} \J{2} + \J{3}) e^{\tau^\ast \i k_{cz} z} \right) \bd{X}^{\ast},
\end{equation}
where
\begin{equation}\label{eq-ewq-X}
\bd{X}^\ast = \delta_{j,t}\bd{X}^{\mathrm{f}\ast} + \bd{X}^{\mathrm{r}\ast}.
\end{equation}

The unknowns $\bd{X}^{\mathrm{r}\ast}$ from all $L+1$ layers should be determined by the interface equations and the radiation equations specified below.

The interface conditions of the Green's function are given by \cref{eq-ewq-if-ss}--\cref{eq-ewq-if-sv} by cases.
One can re-organize them into the representation using matrices $\J{1}, \cdots, \J{5}$ in the following way:
(a) from $\lbr \mathcal{T}_{33} \rbr = 0$ we set up the interface equations
\begin{align}\label{eq-ewq-t33}
& \left \lbr \gamma \frac{\partial G_{3l}}{\partial z} + \lambda \left( \frac{\partial G_{1l}}{\partial x} + \frac{\partial G_{2l}}{\partial y} \right) \right \rbr = 0, \quad l=1,2,3 \\
\Leftrightarrow & \left\lbr (\lambda \J{4} + \gamma \partial_z \J{2}) \widehat{\bd{G}} \right\rbr = \bd{0} \\
\Leftrightarrow & \left\lbr \sum_{\ast \in \{\uparrow, \downarrow\}} \left( 2\mu \tau^\ast \i k_{sz} e^{\tau^\ast \i k_{sz} z} \J{4} + (-\lambda \kr^2 - \gamma k_{cz}^2) e^{\tau^\ast \i k_{cz} z} \J{2} \right) \bd{X}^\ast \right \rbr = \bd{0};
\end{align}
(b) from $\lbr u_3 \rbr = 0$ we set up the interface equations
\begin{align}\label{eq-ewq-u3}
& \lbr G_{3l} \rbr = 0, \quad l=1,2,3 \\
\Leftrightarrow & \left \lbr \J{2} \widehat{\bd{G}} \right \rbr = \bd{0} \\
\Leftrightarrow & \left\lbr \sum_{\ast \in \{\uparrow, \downarrow\}} \left( e^{\tau^\ast \i k_{sz} z} \J{4} + \tau^\ast \i k_{cz} e^{\tau^\ast \i k_{cz} z} \J{2} \right) \bd{X}^\ast \right \rbr = \bd{0};
\end{align}
(c) from $\lbr \mathcal{T}_{31} \rbr = \lbr \mathcal{T}_{32} \rbr = 0$ we set up the interface equations
\begin{align}\label{eq-ewq-t31t32}
& \left \lbr \mu \left( \frac{\partial G_{3l}}{\partial x} + \frac{\partial G_{1l}}{\partial z} \right) \right \rbr = \left \lbr \mu \left( \frac{\partial G_{3l}}{\partial y} + \frac{\partial G_{2l}}{\partial z} \right) \right \rbr = 0, \quad l = 1,2,3 \\
\Leftrightarrow & \left \lbr \mu \left( \partial_z \J{1} + \J{3}\right) \widehat{\bd{G}} \right \rbr = \bd{0} \\
\Leftrightarrow & \left\lbr \sum_{\ast \in \{\uparrow, \downarrow\}} \left( \mu k_{sz}^2 e^{\tau^\ast \i k_{sz} z} \J{1} + \mu e^{\tau^\ast \i k_{sz} z} \J{5} + 2\mu \tau^\ast \i k_{cz} e^{\tau^\ast \i k_{cz} z} \J{3} \right) \bd{X}^\ast \right \rbr = \bd{0};
\end{align}
(d) from $\lbr u_1 \rbr = \lbr u_2 \rbr = 0$ we set up the interface equations
\begin{align}\label{eq-ewq-u1u2}
& \lbr G_{1l} \rbr = \lbr G_{2l} \rbr = 0, \quad l = 1,2,3 \\
\Leftrightarrow & \left \lbr \J{1} \widehat{\bd{G}} \right \rbr = \bd{0} \\
\Leftrightarrow & \left\lbr \sum_{\ast \in \{\uparrow, \downarrow\}} \left( -\tau^\ast \i k_{sz} e^{\tau^\ast \i k_{sz} z} \J{1} + e^{\tau^\ast \i k_{cz} z} \J{3} \right) \bd{X}^\ast \right \rbr = \bd{0}.
\end{align}

The radiation conditions are similarly treated like in the Maxwell's equations.
The limits as $z \to \pm \infty$ are reduced to
\begin{equation}\label{eq-ewq-radiation-inf}
\left( (\i k_{sz,0} \J{1} + \J{4}) e^{-\i k_{sz,0} z} + (-\i k_{cz,0} \J{2} + \J{3}) e^{-\i k_{cz,0} z} \right) \bd{X}_{0}^{\downarrow} \to \bd{0}
\end{equation}
as $z \to \infty$ in the top layer, which can be simplified as $\bd{X}_0^{\mathrm{r}\downarrow} = \bd{0}$, and
\begin{equation}\label{eq-ewq-radiation-minf}
\left( (-\i k_{sz,L} \J{1} + \J{4}) e^{\i k_{sz,L} z} + (\i k_{cz,L} \J{2} + \J{3}) e^{\i k_{cz,L} z} \right) \bd{X}_{L}^{\uparrow} \to \bd{0}
\end{equation}
as $z \to -\infty$ in the bottom layer, which can be simplified as $\bd{X}_L^{\mathrm{r}\uparrow} = \bd{0}$, i.e. propagating directions from $z=\pm\infty$ are prohibited.

Since the above equations are linear equations of $\bd{X}^{\mathrm{r}\ast}$ for all layers with coefficients in $\mathbb{F}_0$, and are sufficient to uniquely determine the Green's function, i.e. the unknowns $\bd{X}^{\mathrm{r}\ast}$ within each layer, by \cref{thm-solution-filtering}, each
\begin{equation}
\bd{X}_t^{\mathrm{r}\ast} \in \mathfrak{R}^0
\end{equation}
has the matrix basis formulation.

Again we try to reinterpret the interface equations and the radiation equations so that they're written in terms of the matrix basis coefficients.
Suppose in each layer $\bd{X}_t^{\mathrm{r}\ast}$ has the matrix basis expansion
\begin{equation}\label{eq-ewq-Xt-expansion}
\bd{X}_t^{\mathrm{r}\ast} = \sum_{l=1}^{5} x_{l,t}^{\mathrm{r}\ast} \J{l},
\end{equation}
then, we can decompose
\begin{align}
\begin{split}
\bd{X}^\ast = \sum_{l=1}^{5} x_{l}^{\ast} \J{l},
\end{split}
\end{align}
where
\begin{align}
\begin{split}
x_{l}^\uparrow &= \delta_{j,t} 1_{\{ z > z' \}} x_{l}^{\mathrm{f} \uparrow} + x_{l}^{\mathrm{r} \uparrow}, \\
x_{l}^\downarrow &= \delta_{j,t} 1_{\{ z < z' \}} x_{l}^{\mathrm{f} \downarrow} + x_{l}^{\mathrm{r} \downarrow}
\end{split}
\end{align}
are the reaction field decomposition of $x_l^\ast$.
Using \cref{eq-ewq-G-X} and the matrix basis representation of $\bd{X}^\ast$, the radiation equations are simply
\begin{equation}
x_{l,0}^\downarrow = x_{l,L}^\uparrow = 0, \quad l = 1,2,3,4,5.
\end{equation}
For the interface equations the groups (a)--(d) are listed below.
(a) $\lbr \mathcal{T}_{33} \rbr = 0$ was equivalent to \cref{eq-ewq-t33} and now interpreted as
\begin{equation}\label{eq-ewq-t33-J}
\lbr T_1 \J{2} + T_2 \J{4} \rbr = \bd{0} \quad \Leftrightarrow \quad \lbr T_1 \rbr = \lbr T_2 \rbr = 0,
\end{equation}
where
\begin{align}
T_1 &= \sum_{\ast \in \{\uparrow, \downarrow \}} -2\mu \tau^\ast \i k_{sz} \kr^2 e^{\tau^\ast \i k_{sz} z} x_3^\ast + (-\lambda \kr^2 - \gamma k_{cz}^2) e^{\tau^\ast \i k_{cz} z} x_2^\ast, \\
T_2 &= \sum_{\ast \in \{\uparrow, \downarrow \}} 2\mu \tau^\ast \i k_{sz} e^{\tau^\ast \i k_{sz} z} (x_1^\ast - \kr^2 x_5^\ast) + (-\lambda \kr^2 - \gamma k_{cz}^2) e^{\tau^\ast \i k_{cz} z} x_4^\ast.
\end{align}
(b) $\lbr u_3 \rbr = 0$ was equivalent to \cref{eq-ewq-u3} and now interpreted as
\begin{equation}\label{eq-ewq-u3-J}
\lbr T_3 \J{2} + T_4 \J{4} \rbr = \bd{0} \quad \Leftrightarrow \quad \lbr T_3 \rbr = \lbr T_4 \rbr = 0,
\end{equation}
where
\begin{align}
T_3 &= \sum_{\ast \in \{\uparrow, \downarrow \}} -\kr^2  e^{\tau^\ast \i k_{sz} z}x_3^\ast + \tau^\ast \i k_{cz} e^{\tau^\ast\i k_{cz} z} x_2^\ast, \\
T_4 &= \sum_{\ast \in \{\uparrow, \downarrow \}} e^{\tau^\ast \i k_{sz} z} (x_1^\ast - \kr^2 x_5^\ast) + \tau^\ast \i k_{cz} e^{\tau^\ast\i k_{cz} z} x_4^\ast.
\end{align}
(c) $\lbr \mathcal{T}_{31} \rbr = \lbr \mathcal{T}_{32} \rbr =0 $ was equivalent to \cref{eq-ewq-t31t32} and now interpreted as
\begin{equation}\label{eq-ewq-t31t32-J}
\lbr T_5 \J{1} + T_6 \J{3} + T_7 \J{5} \rbr = \bd{0} \quad \Leftrightarrow \quad \lbr T_5 \rbr = \lbr T_6 \rbr = \lbr T_7 \rbr = 0,
\end{equation}
where
\begin{align}
T_5 &= \sum_{\ast \in \{\uparrow, \downarrow \}} \mu k_{sz}^2 e^{\tau^\ast \i k_{sz} z}x_1^\ast , \\
T_6 &= \sum_{\ast \in \{\uparrow, \downarrow \}} \mu (k_{sz}^2 - \kr^2) e^{\tau^\ast \i k_{sz} z} x_3^\ast + 2 \mu \tau^\ast \i k_{cz} e^{\tau^\ast\i k_{cz} z} x_2^\ast, \\
T_7 &= \sum_{\ast \in \{\uparrow, \downarrow \}} \mu e^{\tau^\ast \i k_{sz} z}x_1^\ast + \mu (k_{sz}^2 - \kr^2) e^{\tau^\ast \i k_{sz} z} x_5^\ast + 2\mu \tau^\ast \i k_{cz} e^{\tau^\ast\i k_{cz} z} x_4^\ast.
\end{align}
(d) $\lbr u_1 \rbr = \lbr u_2 \rbr = 0$ was equivalent to \cref{eq-ewq-u1u2} and now interpreted as
\begin{equation}\label{eq-ewq-u1u2-J}
\lbr T_8 \J{1} + T_9 \J{3} + T_{10} \J{5} \rbr = \bd{0} \quad \Leftrightarrow \quad \lbr T_8 \rbr = \lbr T_9 \rbr = \lbr T_{10} \rbr = 0,
\end{equation}
where
\begin{align}
T_8 &= \sum_{\ast \in \{\uparrow, \downarrow \}} -\tau^\ast \i k_{sz} e^{\tau^\ast \i k_{sz} z}x_1^\ast, \\
T_9 &= \sum_{\ast \in \{\uparrow, \downarrow \}} -\tau^\ast \i k_{sz} e^{\tau^\ast \i k_{sz} z}x_3^\ast + e^{\tau^\ast\i k_{cz} z} x_2^\ast, \\
T_{10} &= \sum_{\ast \in \{\uparrow, \downarrow \}} -\tau^\ast \i k_{sz} e^{\tau^\ast \i k_{sz} z}x_5^\ast + e^{\tau^\ast\i k_{cz} z} x_4^\ast.
\end{align}

The interface equations can be divided into two groups of continuity equations for solving, depending on the involvement of the unknowns: $x_1^\ast, x_4^\ast$ and $x_5^\ast$ consist the continuity equations of $T_2, T_4, T_5$, $T_7, T_8$ and $T_{10}$, while $x_2^\ast$ and $x_3^\ast$ appear in the other group with members $T_1, T_3, T_6$ and $T_9$.

\begin{remark}[separation of elastic waves in the matrix basis formulation]
When representing $T_1, \cdots, T_{10}$ using $x_1^\ast, \cdots, x_5^{\ast}$, the S-wave parts with $z$-dependence $e^{\pm \i k_{sz}}$ are only related with $x_1^{\ast}, x_3^\ast$ and $x_5^\ast$, and the P-wave parts are only related with $x_2^\ast$ and $x_4^\ast$.
In the matrix basis representation of $\bd{X}^\ast$, $x_1, x_3$ and $x_5$ takes the first and the second row of the matrix together with the basis matrices $\J{1}, \J{3}$ and $\J{5}$, and $x_2, x_4$ takes the last row of $\bd{X}^\ast$ together with the basis matrices $\J{2}$ and $\J{4}$.
Therefore, the rows of $\bd{X}^\ast$ naturally separates the elastic wave into the S-wave part and the P-wave part.
\end{remark}

\subsubsection{Case when there's a vacuum half space}
Now consider the case when a vacuum half space is added to the problem.
Without loss of generality, suppose the top layer is replaced by the vacuum.
For convenience we define $\bd{X}^{\mathrm{r}\ast} = \bd{0}$ in the vacuum side.
The interface equations at $z = d_0 - 0$ are the zero-traction conditions
\begin{equation}
\mathcal{T}_{31} = \mathcal{T}_{32} = 0, \quad \mathcal{T}_{33} = 0,
\end{equation}
which are corresponding to \cref{eq-ewq-t31t32} and \cref{eq-ewq-t33}, respectively, except from the fact that wave propagation does not exist in the vacuum side.
Hence the matrix basis representation of \cref{eq-ewq-t31t32} and \cref{eq-ewq-t33} can be ported for these conditions, resulting in the following equations at $z = d_0 - 0$:
\begin{align}
\sum_{\ast \in \{\uparrow, \downarrow\}} \left( \mu k_{sz}^2 e^{\tau^\ast \i k_{sz} z} \J{1} + \mu e^{\tau^\ast \i k_{sz} z} \J{5} + 2\mu \tau^\ast \i k_{cz} e^{\tau^\ast \i k_{cz} z} \J{3} \right) \bd{X}^\ast = \bd{0},
\end{align}
\begin{align}
\sum_{\ast \in \{\uparrow, \downarrow\}} \left( 2\mu \tau^\ast \i k_{sz} e^{\tau^\ast \i k_{sz} z} \J{4} + (-\lambda \kr^2 - \gamma k_{cz}^2) e^{\tau^\ast \i k_{cz} z} \J{2} \right) \bd{X}^\ast = \bd{0}.
\end{align}
The radiation condition \cref{eq-ewq-radiation-inf} from $z = +\infty$ is no longer needed.
Again we have collected the restricting equations of $\bd{X}_t^{\mathrm{r}\ast}$ for $1 \le t \le L$, and the same theories can be applied.
The case for two half-space vacuum layers is similar.

\subsubsection{Case when there's a fluid layer}
Finally we consider the case when one or more fluid layer is involved.
Our goal is to justify that $\widehat{\bd{G}}$ still has the matrix basis formulation within each layer.

Since we suppose the source is from a solid layer, in any fluid layer $\widehat{\bd{G}} = \widehat{\bd{G}}^\mathrm{r}$.
Because the S-wave does not propagate in fluid media, the formulation \cref{eq-ewq-G-X} can not be directly used.
Therefore we must take a step back to the acoustic wave equation in the fluid
\begin{align}\label{eq-ewq-fluid-acoustic}
\lambda \nabla \nabla \cdot {\bd{G}}^{\mathrm{r}} + \omega^2 \rho {\bd{G}}^{\mathrm{r}} &= \bd{0}, \\
\nabla \times {\bd{G}}^{\mathrm{r}} &= \bd{0}.
\end{align}
In the frequency domain, the general solution of these equations is given by
\begin{equation}\label{eq-ewq-fluid-Gr-general}
\widehat{\bd{G}}^{\mathrm{r}} = \begin{bmatrix}
\i k_x \\ \i k_y \\ \i k_{cz}
\end{bmatrix} \begin{bmatrix}
\psi_1^\uparrow \\
\psi_2^\uparrow \\
\psi_3^\uparrow
\end{bmatrix}^T e^{\i k_{cz} z} + \begin{bmatrix}
\i k_x \\ \i k_y \\ -\i k_{cz}
\end{bmatrix} \begin{bmatrix}
\psi_1^\downarrow \\
\psi_2^\downarrow \\
\psi_3^\downarrow
\end{bmatrix}^T e^{-\i k_{cz} z}
\end{equation}
where each $\psi_{l}^\ast$ does not depend on $z$.
Here we take a weaker form while keeping the $z$ variable separating
\begin{equation}
\widehat{\bd{G}}^{\mathrm{r}} = \widehat{\bd{G}}^{\mathrm{r}\uparrow} e^{\i k_{cz} z} + \widehat{\bd{G}}^{\mathrm{r}\downarrow} e^{-\i k_{cz} z},
\end{equation}
and treat $\widehat{\bd{G}}^{\mathrm{r}\ast}$ as unknowns (which are independent of $z$).
In the frequency domain, the first equation of \cref{eq-ewq-fluid-acoustic} is reinterpreted as
\begin{align}\label{eq-ewq-fluid-acoustic-equiv1}
\begin{split}
{}&{} \lambda \left( -k_{cz}^2 \J{2} + \i k_{cz}\J{3} + \i k_{cz} \J{4} + \J{5} \right)\widehat{\bd{G}}^{\mathrm{r}\uparrow} e^{\i k_{cz} z} \\
+ {}&{} \lambda \left( -k_{cz}^2 \J{2} - \i k_{cz}\J{3} - \i k_{cz} \J{4} + \J{5} \right) \widehat{\bd{G}}^{\mathrm{r}\downarrow} e^{-\i k_{cz} z} = \bd{0},
\end{split}
\end{align}
and the second equation is reinterpreted as $(\J{6} + \J{7} - \partial_z \J{9}) \widehat{\bd{G}}^{\mathrm{r}} = \bd{0}$ and is equivalent to the \emph{pair} of equations
\begin{align}\label{eq-ewq-fluid-acoustic-equiv2}
\begin{split}
(\kr^2 \J{2} + \i k_{cz} \J{4}) \widehat{\bd{G}}^{\mathrm{r}\uparrow} e^{\i k_{cz} z} + (\kr^2 \J{2} - \i k_{cz} \J{4}) \widehat{\bd{G}}^{\mathrm{r}\downarrow} e^{-\i k_{cz} z} = \J{6} (\J{6} + \J{7} - \partial_z \J{9} ) \widehat{\bd{G}}^\mathrm{r} &= \bd{0}, \\
(\kr^2 \J{1} + \J{5}) (\widehat{\bd{G}}^{\mathrm{r}\uparrow} e^{\i k_{cz} z} + \widehat{\bd{G}}^{\mathrm{r}\downarrow} e^{-\i k_{cz} z}) = \J{7} (\J{6} + \J{7} - \partial_z \J{9} ) \widehat{\bd{G}}^\mathrm{r} &= \bd{0},
\end{split}
\end{align}
which are written using only $\J{1}, \cdots, \J{5}$.
The radiation condition, if applied, is simply $\widehat{\bd{G}}^{\mathrm{r} \ast} = \bd{0}$ for the corresponding prohibited propagation direction.
The interface equations are a subset of \cref{eq-ewq-t33}, \cref{eq-ewq-u3}, \cref{eq-ewq-t31t32} and \cref{eq-ewq-u1u2}.
For each of these equations, in the solid side we adopt the representation using unknowns $\bd{X}^{\mathrm{r} \ast}$, and in the fluid side we use the representation with unknowns $\widehat{\bd{G}}^{\mathrm{r}\ast}$.
The resulting series of restricting equations is a linear system of unknowns from each layer with coefficients in $\mathfrak{R}^0$, so \Cref{thm-solution-filtering} can be applied.
Together with the uniqueness of the Green's function, we again conclude that each $\bd{X}^{\mathrm{r} \ast}$ from solid layers and each $\widehat{\bd{G}}^{\mathrm{r} \ast}$ from fluid layers have a matrix basis representation.

\begin{remark}
We can still separate the unknown matrix basis coefficients into groups exactly like the pure solid case, i.e. the equations of $T_2, T_4, T_5$, $T_7, T_8$ and $T_{10}$ consist one group and the rest consist the other one.
\end{remark}

\subsection{The dyadic Green's function in layered media with source in fluid}
When the source is in a fluid layer, the Green's function $g_p$ with respect to the pressure is a scalar function, and the corresponding $\gu$ with respect to the displacement is a vector function, hence no $3 \times 3$ tensor seems involved.
However, we can still simplify the formulation of $\gu$ by using a vector version derived from the matrix basis \cref{mat-J}.

\subsubsection{The vector basis}
Define $3 \times 1$ vectors
\begin{equation}
\bd{j}_2 = \begin{bmatrix}
0 \\ 0 \\ 1
\end{bmatrix}, \quad \bd{j}_3 = \begin{bmatrix}
\i k_x \\ \i k_y \\ 0
\end{bmatrix}, \quad \bd{j}_7 = \begin{bmatrix}
\i k_y \\ -\i k_x \\ 0
\end{bmatrix},
\end{equation}
which are the third column of $\J{2}, \J{3}$ and $\J{7}$, respectively.
The vectors $\bd{j}_2, \bd{j}_3$ and $\bd{j}_7$ exactly consist all the non-trivial third columns of the matrix basis $\J{1}, \cdots, \J{9}$.
Hence the product between a $\bd{J}_i$ matrix and a $\bd{j}_j$ vector
\begin{equation}
\bd{J}_i \cdot \bd{j}_j = \J{i} \cdot \bd{J}_j \bd{e}_3 = (\bd{J}_i \bd{J}_j) \bd{e}_3
\end{equation}
is still a linear combination of $\bd{j}_2, \bd{j}_3$ and $\bd{j}_7$ (with coefficients $\pm 1, \pm \kr^2$).

Similar to the linear spaces of block matrices defined in \cref{eq-block-matrices}, for any $p, q \in \mathbb{N}$ and any subfield $\mathbb{K} \subset \mathbb{F}$, define
\begin{align}\label{eq-ewq-block-vectors}
\begin{split}
\mathfrak{r}_{p\times q}(\mathbb{K}) &= \left\{ \bd{K}_2 \otimes \bd{j}_2 + \bd{K}_3 \otimes \bd{j}_3: \bd{K}_2, \bd{K}_3 \in \mathbb{K}^{p \times q} \right\}, \\
\mathfrak{i}_{p\times q}(\mathbb{K}) &= \left\{ \bd{K}_7 \otimes \bd{j}_7: \bd{K}_7 \in \mathbb{K}^{p \times q} \right\}, \\
\mathfrak{m}_{p \times q}(\mathbb{K}) &= \mathfrak{r}_{p\times q}(\mathbb{K}) \oplus \mathfrak{i}_{p\times q}(\mathbb{K}).
\end{split}
\end{align}
Note that $\mathfrak{i}_{p\times q}(\mathbb{F}) = \mathbb{F}^{3p \times q}$ because $\bd{j}_2, \bd{j}_3$ and $\bd{j}_7$ are linearly independent.
The solution filtering \Cref{thm-solution-filtering} is accordingly adjusted as follows.
\begin{theorem}[Vector solution filtering] \label{thm-vector-solution-filtering}
Suppose $p,q,r \in \mathbb{N}$, the block matrices $\bar{\bd{A}} \in \mathfrak{R}_{p \times r}(\mathbb{F}_0)$, $\bar{\bd{X}} \in \mathfrak{m}_{r \times q}(\mathbb{F})$ and $\bar{\bd{B}} \in \mathfrak{r}_{p \times q}(\mathbb{F}_0)$ satisfy $\bar{\bd{A}} \cdot \bar{\bd{X}} = \bar{\bd{B}}$.
Then, there exists a ``filtered'' block matrix $\bar{\bd{X}}_0 \in \mathfrak{r}_{r \times q}(\mathbb{F}_0)$ such that $\bar{\bd{A}} \cdot \bar{\bd{X}}_0 = \bar{\bd{B}}$.
\end{theorem}
\begin{proof}
For any $u, v \in \mathbb{N}$, the mapping
\begin{equation}
\mathcal{P}_{u \times v}: \bd{K}_2 \otimes \bd{j}_2 + \bd{K}_3 \otimes \bd{j}_3 + \bd{K}_7 \otimes \bd{j}_7 \in \mathfrak{m}_{u \times v}(\mathbb{K}) \mapsto \bd{K}_2 \otimes  \J{2} + \bd{K} \otimes  \J{3} + \bd{K}_7 \otimes  \J{7} \in \mathfrak{M}_{u \times v}(\mathbb{K})
\end{equation}
is bijective.
It is obvious that
\begin{equation}
\bar{\bd{A}} \cdot \mathcal{P}_{r \times q}(\bar{\bd{X}}) = \mathcal{P}_{p \times q}(\bar{\bd{B}})
\end{equation}
and that $\mathcal{P}_{p \times q}(\bar{\bd{B}}) \in \mathfrak{R}_{p \times q}(\mathbb{F}_0)$.
By \Cref{thm-solution-filtering}, there exists $\bar{\bd{X}}_M \in \mathfrak{R}_{r \times q}(\mathbb{F}_0)$ such that
\begin{equation}
\bar{\bd{A}} \cdot \bar{\bd{X}}_M = \mathcal{P}_{p \times q}(\bar{\bd{B}}).
\end{equation}
Then, by removing the first and the second column of each $3 \times 3$ block from $\bar{\bd{X}}_M$ we get the desired result.
\end{proof}
\subsubsection{Derivation of the vector basis formulation}
We will first take less words to repeat the reaction field decomposition.
The free space Green's function $\gu^{\mathrm{f}}$ in the fluid media with respect to the displacement $\bd{u}$ was given in \cref{eq-ewq-guf}.
In the frequency domain,
\begin{align}\label{eq-ewq-guf-freq}
\begin{split}
\widehat{\bd{g}}_{\bd{u}}^{\mathrm{f}} &= \frac{1}{\omega^2 \rho_j} \begin{bmatrix}
\i k_x \\ \i k_y \\ \partial_z
\end{bmatrix} \widehat{g}_p^{\mathrm{f}} = \frac{1}{\omega^2 \rho_j} \begin{bmatrix}
\i k_x \\ \i k_y \\ \partial_z
\end{bmatrix} \frac{\i e^{\i k_{cz,j} |z-z'|}}{2 k_{cz,j}} \\
&= e^{\i k_{cz,j}z} \guh^{\mathrm{f}\uparrow} + e^{-\i k_{cz,j}z} \guh^{\mathrm{f}\downarrow}, \\
&= (\i k_{cz,j} \bd{j}_2 + \bd{j}_3) e^{\i k_{cz,j}} g_{\bd{u}}^{\mathrm{f}\uparrow} + (-\i k_{cz,j} \bd{j}_2 + \bd{j}_3) e^{-\i k_{cz,j}} g_{\bd{u}}^{\mathrm{f}\downarrow},
\end{split}
\end{align}
where
\begin{equation}
\guh^{\mathrm{f}\ast} = (\tau^\ast\i k_{cz,j} \bd{j}_2 + \bd{j}_3) g_{\bd{u}}^{\mathrm{f}\ast}, \quad
g_{\bd{u}}^{\mathrm{f}\uparrow} = 1_{\{z > z'\}}\frac{\i e^{-\i k_{cz,j}z'} }{2\omega^2 \rho_j k_{cz,j}}, \quad
g_{\bd{u}}^{\mathrm{f}\downarrow} = 1_{\{z < z'\}} \frac{\i e^{\i k_{cz,j}z'} }{2\omega^2 \rho_j k_{cz,j}}.
\end{equation}

The reaction field of the Green's function is defined by
\begin{equation}
\gu^\mathrm{r} = {\bd{g}}_{\bd{u}} - \delta_{j,t} {\bd{g}}_{\bd{u}}^{\mathrm{f}}
\end{equation}
in each layer.
In a fluid layer, the reaction field Green's function $\gu^{\mathrm{r}}$ satisfies the acoustic wave equations
\begin{align}
\begin{split}
\lambda \nabla \nabla \cdot \gu^{\mathrm{r}} + \omega^2 \rho \gu^{\mathrm{r}} &= \vec{0}, \\
\nabla \times \gu^{\mathrm{r}} &= \vec{0},
\end{split}
\end{align}
like in \cref{eq-ewq-fluid-acoustic}.
When treated as ordinary differential equations of $z$ in the frequency domain, the general solutions to these equations are
\begin{equation}\label{eq-ewq-fluid-curlfree-general}
\gur = (\i k_{cz} \bd{j}_2 + \bd{j}_3) e^{\i k_{cz} z} \psi^{\uparrow} + (-\i k_{cz} \bd{j}_2 + \bd{j}_3) e^{-\i k_{cz} z} \psi^{\downarrow},
\end{equation}
where the scalars $\psi^{\ast}$ do not depend on $z$.
We take the weaker form
\begin{equation}
\gur = e^{\i k_{cz} z} \guh^{\mathrm{r}\uparrow} + e^{-\i k_{cz} z} \guh^{\mathrm{r}\downarrow}
\end{equation}
where $\guh^{\mathrm{r}\ast} = g_{\bd{u}}^{\mathrm{r} \ast} (\tau^\ast \i k_{cz} \bd{j}_2 + \bd{j}_3)$, satisfying the equations
\begin{align}\label{eq-ewq-fluid-acoustic-vector-equiv1}
\begin{split}
{}&{} \lambda \left( -k_{cz}^2 \J{2} + \i k_{cz}\J{3} + \i k_{cz} \J{4} + \J{5} \right)\guh^{\mathrm{r}\uparrow} e^{\i k_{cz} z} \\
+ {}&{} \lambda \left( -k_{cz}^2 \J{2} - \i k_{cz}\J{3} - \i k_{cz} \J{4} + \J{5} \right) \guh^{\mathrm{r}\downarrow} e^{-\i k_{cz} z} = \vec{0},
\end{split}
\end{align}
and
\begin{align}\label{eq-ewq-fluid-acoustic-vector-equiv2}
\begin{split}
(\kr^2 \J{2} + \i k_{cz} \J{4}) \guh^{\mathrm{r}\uparrow} e^{\i k_{cz} z} + (\kr^2 \J{2} - \i k_{cz} \J{4}) \guh^{\mathrm{r}\downarrow} e^{-\i k_{cz} z} = \J{6} (\J{6} + \J{7} - \partial_z \J{9} ) \guh^\mathrm{r} &= \vec{0}, \\
(\kr^2 \J{1} + \J{5}) (\guh^{\mathrm{r}\uparrow} e^{\i k_{cz} z} + \guh^{\mathrm{r}\downarrow} e^{-\i k_{cz} z}) = \J{7} (\J{6} + \J{7} - \partial_z \J{9} ) \guh^\mathrm{r} &= \vec{0},
\end{split}
\end{align}
which are the vector versions of \cref{eq-ewq-fluid-acoustic-equiv1} and \cref{eq-ewq-fluid-acoustic-equiv2}, respectively.
Then, let
\begin{equation}
\guh^{\ast} = \delta_{j,t} \guh^{\mathrm{f} \ast} + \guh^{\mathrm{r}\ast}.
\end{equation}
One can verify that $\guh^{\ast}$ satisfies the same equations as the above of $\guh^{\mathrm{r}\ast}$ within each layer when provided $z \ne z'$.
In a solid layer, the reaction field Green's function $\gu^\mathrm{r} = \gu$ satisfies the homogeneous elastic wave equation
\begin{equation}
(\lambda+\mu)\nabla \nabla \cdot \gu^\mathrm{r} + \mu \nabla^2 \gu^\mathrm{r} + \omega^2 \rho \gu^\mathrm{r} = \bd{0}
\end{equation}
like in \cref{eq-ewq-Gr}.
In the frequency domain we again have the general solution
\begin{align}\label{eq-homo-ewq-vector-general}
\begin{split}
\gur = \sum_{\ast \in \{ \uparrow, \downarrow \}} \left( (-\tau^\ast \i k_{sz} \J{1} + \J{4}) e^{\tau^\ast\i k_{sz} z} + (\tau^\ast \i k_{cz} \J{2} + \J{3}) e^{\tau^\ast \i k_{cz} z} \right) \bd{x}^{\mathrm{r}\ast},
\end{split}
\end{align}
where $\bd{x}^{\mathrm{r}\ast}$ are $3\times 1$ vectors within each layer independent from $z$.
This repeats the result of \cref{eq-homo-ewq-general}.

The interface equations and the radiation equations are interpreted as the vector versions of the tensor case we've discussed in \Cref{section-Gf-source-in-solid}, just by replacing the unknowns from tensors to vectors, and replacing the free-space contribution from $ \widehat{\bd{G}}^{\mathrm{f}\ast}$ to $\guh^{\mathrm{f}\ast}$.
Therefore, by \Cref{thm-vector-solution-filtering}, the solution to the unknown vectors, which is known uniquely determined by the restricting equations (since the coefficients of unknowns are shared with the tensor version), has the vector basis representation, i.e. each $\guh^{\mathrm{r}\ast} \in \mathfrak{r}_{1 \times 1}(\mathbb{F}_0)$ and each $\bd{x}^{\mathrm{r}\ast} \in \mathfrak{r}_{1 \times 1}(\mathbb{F}_0)$.

Finally we use the vector basis to reinterpret the interface equations.
Suppose in the fluid layers
\begin{equation}
\guh^{\mathrm{r}\ast} = g_{\bd{u},2}^{\mathrm{r}\ast} \bd{j}_2 + g_{\bd{u},3}^{\mathrm{r}\ast} \bd{j}_3.
\end{equation}
Since the displacement is curl-free, the general solution form \cref{eq-ewq-fluid-curlfree-general} implies
\begin{equation}
g_{\bd{u},2}^{\mathrm{r}\ast} = \tau^\ast \i k_{cz} g_{\bd{u},3}^{\mathrm{r}\ast}.
\end{equation}
Hence we have the reaction field decomposition in fluid layers
\begin{equation}
\guh = (\i k_{cz} \bd{j}_2 + \bd{j}_3) e^{\i k_{cz}} g_{\bd{u}}^{\uparrow} + (-\i k_{cz} \bd{j}_2 + \bd{j}_3) e^{-\i k_{cz}} g_{\bd{u}}^{\downarrow},
\end{equation}
where
\begin{equation}
g_{\bd{u}}^{\ast} = \delta_{j,t} g_{\bd{u}}^{\mathrm{f}\ast} + g_{\bd{u},3}^{\mathrm{r} \ast}.
\end{equation}
Then suppose in the solid layers
\begin{equation}
\bd{x}^{\mathrm{r}\ast} = x_{2}^{\ast} \bd{j}_2 + x_{3}^{\ast} \bd{j}_3.
\end{equation}
The general solution in solid layers \cref{eq-homo-ewq-vector-general} is then simplified as
\begin{equation}
\guh = \widehat{\bd{g}}_{\bd{u}}^{\mathrm{r}} = \sum_{\ast \in \{\uparrow,\downarrow\}}
\left( (-\tau^\ast \i k_{sz} \bd{j}_3 - \kr^2 \bd{j}_2) x_3^\ast e^{\tau^\ast \i k_{sz} z}
+ (\tau^\ast \i k_{cz} \bd{j}_2 + \bd{j}_3) x_2^\ast e^{\tau^\ast \i k_{cz} z} \right),
\end{equation}
i.e. rows of $\bd{x}^{\mathrm{r}\ast}$ again reveals the wave decomposition in the solid media, with variable $x_2^\ast$ for the S-wave and $x_3^\ast$ for the P-wave.
For the simplification of the interface equations and the radiation equations with the vector basis $\bd{j}_2$ and $\bd{j}_3$ and their coefficients, it suffices to repeat the matrix forms used in \cref{eq-ewq-t33}, \cref{eq-ewq-u3}, \cref{eq-ewq-t31t32} and \cref{eq-ewq-u1u2}.
(a) $\lbr \mathcal{T}_{33} \rbr = 0$ is equivalent to
\begin{align}
\left\lbr (\lambda \J{4} + \gamma \partial_z \J{2}) \guh \right\rbr = \vec{0}.
\end{align}
In solid it's expanded as the terms in the brackets of \cref{eq-ewq-t33-J} multiplied by $\bd{j}_2$
\begin{equation}
(\lambda \J{4} + \gamma \partial_z \J{2}) \guh
= \left( \sum_{\ast \in \{\uparrow, \downarrow \}} -2\mu \tau^\ast \i k_{sz} \kr^2 e^{\tau^\ast \i k_{sz} z} x_3^\ast + (-\lambda \kr^2 - \gamma k_{cz}^2) e^{\tau^\ast \i k_{cz} z} x_2^\ast \right) \bd{j}_2,
\end{equation}
while in fluid
\begin{equation}
(\lambda \J{4} + \gamma \partial_z \J{2}) \guh
= \left( \sum_{\ast \in \{\uparrow, \downarrow \}} -\omega^2 \rho e^{\tau^\ast \i k_{cz} z}g_{\bd{u}}^\ast \right) \bd{j}_2.
\end{equation}
(b) $\lbr u_3 \rbr = 0$ is equivalent to
\begin{equation}
\left\lbr \J{2} \guh \right\rbr = \vec{0}.
\end{equation}
In solid the terms inside the brackets are expanded as
\begin{equation}
\J{2} \guh = \left( \sum_{\ast \in \{\uparrow, \downarrow \}} -\kr^2  e^{\tau^\ast \i k_{sz} z}x_3^\ast + \tau^\ast \i k_{cz} e^{\tau^\ast\i k_{cz} z} x_2^\ast \right) \bd{j}_2
\end{equation}
and in fluid
\begin{equation}
\J{2} \guh = \left( \sum_{\ast \in \{\uparrow, \downarrow \}} \tau^\ast \i k_{cz} e^{\tau^\ast\i k_{cz} z} g_{\bd{u}}^\ast \right) \bd{j}_2.
\end{equation}
(c) $\lbr \mathcal{T}_{31} \rbr = \lbr \mathcal{T}_{32} \rbr = 0$ is equivalent to
\begin{equation}
\left \lbr \mu \left( \partial_z \J{1} + \J{3}\right) \guh \right \rbr = \vec{0}
\end{equation}
which does not influence the fluid, and the solid side should satisfy
\begin{equation}
\mu \left( \partial_z \J{1} + \J{3}\right) \guh = \left(\sum_{\ast \in \{\uparrow, \downarrow \}} \mu (k_{sz}^2 - \kr^2) e^{\tau^\ast \i k_{sz} z} x_3^\ast + 2 \mu \tau^\ast \i k_{cz} e^{\tau^\ast\i k_{cz} z} x_2^\ast \right) \bd{j}_3 = \vec{0}.
\end{equation}
(d) $\lbr u_1 \rbr = \lbr u_2 \rbr = 0$ is equivalent to
\begin{equation}
\left \lbr \J{1} \guh \right \rbr = \vec{0}.
\end{equation}
In solid we have
\begin{equation}
\J{1} \guh = \left (\sum_{\ast \in \{\uparrow, \downarrow \}} -\tau^\ast \i k_{sz} e^{\tau^\ast \i k_{sz} z}x_3^\ast + e^{\tau^\ast\i k_{cz} z} x_2^\ast \right) \bd{j}_3
\end{equation}
and in fluid,
\begin{equation}
\J{1} \guh = \left (\sum_{\ast \in \{\uparrow, \downarrow \}} e^{\tau^\ast\i k_{cz} z} g_{\bd{u}}^\ast \right) \bd{j}_3.
\end{equation}
The above equations are exactly the variants of $T_1, T_3, T_6$ and $T_9$ defined before.
The 6-term group of unknowns previously in the solid-source case is gone.

\section{Conclusion}
In this paper, a matrix basis formulation is proposed for handling the dyadic Green's functions of the Maxwell's equations and of the elastic wave equation in layered media.
The formulation is then used to further simplify the representation and derivation of the Green's functions in both cases.
As a corollary, a vector basis formulation is introduced to handle the dyadic Green's function of the elastic wave equation with source in fluid media.

\end{document}